\documentclass{article}
\usepackage{arxiv}

\usepackage[utf8]{inputenc}

\usepackage{amssymb,amsmath,amsthm}
\newtheorem{theorem}{Theorem}
\newtheorem{lemma}{Lemma}
\newtheorem{proposition}{Proposition}
\newtheorem{definition}{Definition}

\usepackage{subcaption}
\usepackage{color}
\usepackage{mathtools}
\usepackage{amsfonts}
\usepackage{arydshln}
\usepackage{graphicx}
\usepackage{tikz}
\usepackage{tikz-3dplot}
\usetikzlibrary{calc}
\DeclareUnicodeCharacter{2212}{-}

\usepackage{algorithm}
\usepackage[noend]{algpseudocode}

\algrenewcommand\algorithmicrequire{\textbf{Input:}}
\algrenewcommand\algorithmicensure{\textbf{Output:}}

\usepackage{hyperref}
\usepackage{afterpage}

\newcounter{rownumbers}

\title{A New Geometric Representation for 3D Bijective Mappings and Applications}

\author{\hspace{1mm}Qiguang Chen\\
	Department of Mathematics\\
    The Chinese University of Hong Kong\\
	\texttt{qgchen@math.cuhk.edu.hk} \\
	\And
	\hspace{1mm}Lok Ming Lui \\
	Department of Mathematics\\
    The Chinese University of Hong Kong\\
	\texttt{lmlui@math.cuhk.edu.hk} \\
}

\begin{document}
\maketitle

\begin{abstract}
Three-dimensional (3D) mappings are fundamental in various scientific and engineering applications, including computer-aided engineering (CAE), computer graphics, and medical imaging. They are typically represented and stored as three-dimensional coordinates to which each vertex is mapped. With this representation, manipulating 3D mappings while preserving desired properties becomes challenging. In this work, we present a novel geometric representation for 3D bijective mappings, termed {\it 3D quasiconformality (3DQC)}, which generalizes the concept of Beltrami coefficients from 2D to 3D spaces. This geometric representation facilitates the scientific computation of 3D mapping problems by capturing local geometric properties in 3D mappings. We derive a partial differential equation (PDE) that links the 3DQC to its corresponding mapping. This PDE is discretized into a symmetric positive-definite linear system, which can be efficiently solved using the conjugate gradient method. 3DQC offers a powerful tool for manipulating 3D mappings while maintaining their desired geometric properties. Leveraging 3DQC, we develop numerical algorithms for sparse modeling and numerical interpolation of bijective 3D mappings, facilitating the efficient processing, storage, and manipulation of complex 3D mappings while ensuring bijectivity. Extensive numerical experiments validate the effectiveness and robustness of our proposed methods.
    
\end{abstract}


\section{Introduction}\label{sec:intro}

Mapping data can be seen in various scientific computing and engineering scenarios. For instance, in scientific computing, the results of domain parameterization is represented as mappings that transform complex domains into simpler parameter domains \cite{choi2015flash,lyu2023two}, thereby enabling more efficient computations. In finite element analysis, deforming a mesh is crucial for solving partial differential equations (PDEs) in dynamic domains, and the deformation is represented as a mapping. In image registration task, where the objective is to align different images into a common coordinate system, the alignment is a mapping \cite{zhang2015variational,zhang20203d}. Additionally, topology-preserving image segmentation can be regarded as a task to seek a bijective mapping that transforms a template mask into a segmentation mask \cite{DaopingZhang2021Taci,Zhang2025star}. In medical image analysis, computing deformations of anatomical structures is vital for monitoring disease progression, as it allows clinicians to assess changes in the shape and size of organs over time \cite{CHAN2016177,guo2023automatic}. These deformations can also be represented as mappings between anatomical structures. 
 
One of the major challenges commonly encountered in the scientific computation of manipulating mapping data is finding a geometric representation of mappings that is intuitive and easy to manipulate. In 2D cases, quasiconformal theories have been widely applied to address these challenges. In particular, the Beltrami coefficient (BC) has been utilized to represent 2D bijective mappings. The BC captures local geometric distortion under the corresponding 2D bijective mappings, providing a powerful tool for analyzing and manipulating these mappings. From the BC, we can reconstruct the corresponding mapping by solving Beltrami's equation. Using BCs to model the space of orientation-preserving bijective mappings offers several significant advantages. First, it allows for a sparse representation of bijective mappings. By applying the Fourier or wavelet transform to the Beltrami coefficients, a bijective map can be approximated using a truncated set of Fourier or wavelet coefficients. This sparsity is essential for reducing the computational complexity and improving the efficiency of scientific computing problems involving bijective mappings. Secondly, since BCs are geometric quantities that measure local geometric distortion under a mapping, the geometry can be effectively controlled by formulating the task as an optimization problem over the space of BCs. Thirdly, the use of BCs facilitates interpolation in the space of bijective mappings. Specifically, a homotopy that interpolates between two bijective mappings while preserving the bijectivity property can be computed by interpolating between the corresponding BCs. Interpolation is a fundamental operation in scientific computing for problems involving mappings. For instance, mapping interpolation is essential for dynamic mesh generation in numerical simulations on moving interfaces. Moreover, mapping interpolation plays a vital role in optimization problems over mapping spaces, providing the necessary continuity and regularity for gradient-based optimization methods. The Beltrami coefficient (BC) representation has proven to be an effective tool for characterizing and computing 2D mappings. However, this framework is inherently limited to two-dimensional mappings. 

To overcome this limitation, we propose a novel 3D geometric representation for 3D mappings that generalizes the BC concept to three dimensions. Specifically, the following three challenges related to the scientific computation involving 3D mappings will be addressed:

\begin{enumerate}
    \item \textbf{Geometric Representation of 3D bijective mappings:} The objective is to define an operator $\mathcal{R}: \mathcal{M} \to \mathcal{Q}$, where $\mathcal{M}$ represents the space of 3D bijective mappings and $\mathcal{Q}$ represents the space of geometric representations. For every $f \in \mathcal{M}$, the operator $\mathcal{R}(f)$ should effectively capture the geometric information of $f$. Additionally, $\mathcal{R}$ must be reconstructible. That is, there should exist an operator $\mathcal{F}: \mathcal{Q} \to \mathcal{M}$ such that: $\mathcal{F} \circ \mathcal{R}(f) = f, \quad \forall f \in \mathcal{M}$. This ensures a two-way correspondence between the spaces $\mathcal{M}$ and $\mathcal{Q}$, allowing smooth transitions between 3D bijective mappings and their geometric representations. To solve this problem, we propose a novel 3D quasiconformal representation that describes the local geometric distortion induced by a 3D mapping. The representation is reconstructible. Specifically, we derive a partial differential equation (PDE) that establishes the relationship between the proposed 3D quasiconformal representation and its corresponding mapping. By discretizing this PDE, we obtain a symmetric positive definite linear system, which allows the 3D mapping to be efficiently reconstructed from its representation.
    \item \textbf{Interpolation in 3D Bijective mappings:} The objective is to construct a homotopy $\mathcal{H} : \Omega_1 \times [0,1] \to \Omega_2$ between two 3D bijective mappings, $f_1 : \Omega_1 \to \Omega_2$ and $f_2 : \Omega_1 \to \Omega_2$, such that $\mathcal{H}(\cdot, 0) = f_1$, $\mathcal{H}(\cdot, 1) = f_2$, and $\mathcal{H}(\cdot, t)$ remains a bijection for all $t \in [0,1]$. In this paper, we propose an algorithm to compute $\mathcal{H}$ using an interpolation process over the space of 3D quasiconformal representations. By preserving certain properties of 3D quasiconformal representations, we ensure the construction of a continuous flow of bijective mappings between $f_1$ and $f_2$.
    \item \textbf{Sparse Modelling in 3D Bijective mappings:} The objective is to derive a sparse approximation of a 3D bijective mapping. In this work, we propose an algorithm for sparse modelling of 3D bijective mappings, leveraging the eigen-decomposition of 3D quasiconformal representations via the Laplace-Beltrami operator. High-frequency components are truncated, retaining only the low-frequency components. The proposed algorithm achieves a sparse approximation of 3D bijective mappings with high accuracy while preserving bijectivity.

\end{enumerate}

\smallskip

Extensive experiments have been carried out to test the performance of our proposed algorithms. Results demonstrate the effectiveness of our proposed methods. The remainder of the paper is organized as follows: Section~\ref{sec:contributions} outlines the main contributions of this work. Section~\ref{sec:related_works} provides an overview of related works. In section~\ref{sec:background}, we present the mathematical background necessary for understanding the proposed approach. Section~\ref{sec:continuous} introduces the 3D quasiconformal representation of 3D mappings, and a partial differential equation that connects the representation and its corresponding mapping, after which the discretized 3D quasiconformal representation and an efficient reconstruction algorithm are proposed in section~\ref{sec:discrete}. Section~\ref{sec:analysis} provides a detailed analysis of the proposed framework. Applications of the proposed framework are discussed in section~\ref{sec:app}, and the experimental results of the reconstruction algorithm and applications are shown in section~\ref{sec:experiments}. Finally, the paper is concluded in section~\ref{sec:conclusion}.

\section{Contributions}\label{sec:contributions} 
The main contributions of this paper are outlined as follows. 

\medskip

\begin{enumerate}
    \item We propose an effective 3D quasiconformal representation for 3D bijective mappings, which quantifies local geometric distortions.
    \item We systematically investigate the relationship between the 3D quasiconformal representation and its corresponding 3D mapping. Specifically, we derive a partial differential equation (PDE) that establishes a connection between the proposed 3D quasiconformal representation and its associated mapping.
    \item By discretizing the derived PDE, we develop a reconstruction scheme to recover the 3D bijective mapping from its corresponding 3D quasiconformal representation. The properties of the discretized linear system are rigorously analyzed to identify a suitable numerical scheme for efficient solution.
    \item Building on the proposed framework, we propose a numerical interpolation method for 3D bijective mappings.
    \item Additionally, we introduce a sparse modeling framework for 3D bijective mappings based on the developed algorithm.
    
\end{enumerate}

\section{Related works}\label{sec:related_works}
The computational method of conformal mapping has been widely used in computer graphics, with early contributions from Desbrun et al. \cite{desbrun2002intrinsic} and Lévy et al. \cite{levy2002least}. Gu et al. \cite{gu2004genus} introduced the use of conformal geometry frameworks for surface processing tasks, which has led to significant developments in the field. For example, Zayer, Rössl, and Seidel \cite{zayer2005discrete} presented implicit generalizations of these ideas. Lui and his coauthors extended the framework to quasi-conformal mappings, proposing tools like mapping compression \cite{lui2010compression} and the Linear Beltrami Solver (LBS) \cite{lui2013texture}. The LBS method has been applied to fast spherical and disk conformal parameterization techniques \cite{choi2015flash, choi2015fast} and various registration algorithms for different applications \cite{chen2019image, lam2014genus, lam2015landmark, lam2014landmark, lam2015quasi, lui2014teichmuller, lui2014geometric, lui2010shape, lui2012optimization, qiu2020inconsistent, wen2015landmark, yung2018efficient, zeng2014surface, zhang2019new, zhang2014automatic, zhang2018novel}. It has also been used in image segmentation \cite{chan2018topology, siu2020image, DaopingZhang2021Taci, zhang2021topology_3d} and shape analysis \cite{CHAN2016177, chan2020quasi, choi2020tooth, choi2020shape, lui2013shape, meng2016tempo}. More recently, learning-based methods based on LBS \cite{chen2021deep, guo2023automatic} have been proposed for topology-preserving registration in 2D space.

Despite these advances, extending quasi-conformal methods from 2D to 3D data has not been straightforward, mainly due to limitations in the underlying theory. Some efforts have been made to extend 2D quasi-conformal theory to higher dimensions. Lee et al. \cite{lee2016landmark} defined a generalized conformality distortion quantity to measure the dilation of $n$-dimensional quasi-conformal maps and proposed optimization algorithms for landmark-matching transformations. Zhang et al. \cite{zhang2022unifying} developed a general framework for $n$-dimensional mapping problems using this quantity, and Huang et al. \cite{huang2024topology} proposed a measure consistent with the classical 2D Beltrami coefficient. While these quantities are effective for measuring and controlling mapping distortions, they lack the geometric information needed to reconstruct 3D mappings. As a result, they cannot fully serve as a 3D extension of the Beltrami coefficient.

This work aims to address this limitation by extending the concept of quasi-conformal mapping from 2D to 3D spaces. Specifically, we introduce a geometric representation, called 3D quasiconformality (3DQC), which generalizes the Beltrami coefficient to 3D mappings. This representation enables efficient manipulation and processing of 3D mappings while preserving their geometric properties.

\section{Mathematical background}\label{sec:background}
In this section, we provide a brief review of the mathematical background relevant to this work.

\subsection{Polar decomposition}

In this paper, the left polar decomposition for real rectangular matrices is utilized. The set of $n \times n$ real matrices is denoted as $\mathcal{M}_n$. The left polar decomposition of a matrix $J \in \mathcal{M}_n$ is a factorization of the form $J = UP$, where $U \in \mathcal{M}_n$ is an orthogonal matrix, and $P \in \mathcal{M}_n$ is a symmetric positive semi-definite matrix. Specifically, $P$ can be expressed as $P = \sqrt{J^T J}$.

To understand this decomposition, note that $J^T J$ is a symmetric positive semi-definite matrix. Using eigen decomposition, $J^T J$ can be factorized as  
\[
    J^T J = WDW^{-1},
\]  
where $W$ is an orthogonal matrix and $D$ is a diagonal matrix containing the eigenvalues of $J^T J$. Since $J^T J$ is symmetric, its eigenvalues are non-negative, and $\Sigma$ is defined as the diagonal matrix whose entries are the square roots of the corresponding entries in $D$. These diagonal entries of $\Sigma$ are the singular values of $J$. Thus,  
\[
    P = \sqrt{J^T J} = W\Sigma W^{-1}.
\]

This shows that any linear transformation from $\mathbb{R}^n$ to $\mathbb{R}^n$ can be expressed as a composition of a rotation/reflection and a dilation.

\subsection{2D Quasiconformal maps and Beltrami coefficients}
A surface with a conformal structure is known as a \emph{Riemann surface}. Given two Riemann surfaces $M$ and $N$, a map $f:M \rightarrow N$ is conformal if it preserves the surface metric up to a multiplicative factor called the conformal factor. Quasiconformal maps generalize conformal maps and are orientation-preserving diffeomorphisms between Riemann surfaces with bounded conformal distortion. This means their first-order approximation maps small circles to small ellipses with bounded eccentricity, as illustrated in Figure~\ref{fig:qcmapping}.

\begin{figure}
    \centering
    \begin{tikzpicture}[>=stealth]

  \begin{scope}[shift={(0,0)}]
    \draw[->] (-2,0) -- (2,0) node[below] {$x$};
    \draw[->] (0,-2) -- (0,2) node[left] {$y$};
    \draw[fill=blue!20,opacity=0.5] (0,0) circle (1);
    \fill (0,0) circle (0pt) node[below left, font=\scriptsize] {$O$};
    \draw (0,1) -- ++(-0.1,0) node[left, font=\scriptsize] {$1$};
    \draw (1,0) -- ++(0,-0.1) node[below, font=\scriptsize] {$1$};
    \begin{scope}[rotate=30]
        \draw[->, red] (0,0) -- (1,0) node[pos=0.8, above right, font=\scriptsize, black] {$\begin{pmatrix}\cos{\frac{\theta}{2}}\\ \sin{\frac{\theta}{2}}\end{pmatrix}$};
      \draw[->, blue] (0,0) -- (0,-1) node[pos=0.7, below right, font=\scriptsize, black] {$\begin{pmatrix}\sin{\frac{\theta}{2}}\\-\cos{\frac{\theta}{2}}\end{pmatrix}$};
    \end{scope}
  \end{scope}

  \begin{scope}[shift={(3,0)}]
    \draw[->, thick] (0,0) -- (1,0) node[midway, above] {$S$};
  \end{scope}

  \begin{scope}[shift={(7,0)}]
    \draw[->] (-2,0) -- (2,0) node[below] {$x$};
    \draw[->] (0,-2) -- (0,2) node[left] {$y$};
    \fill (0,0) circle (0pt) node[below left, font=\scriptsize] {$O$};
    \draw (0,1) -- ++(-0.1,0) node[left, font=\scriptsize] {$1$};
    \draw (1,0) -- ++(0,-0.1) node[below, font=\scriptsize] {$1$};
    \begin{scope}[rotate=30]
      \draw[fill=green!20,opacity=0.5] (0,0) ellipse (1.5 and 0.5);
      \draw[->, red] (0,0) -- (1.5,0) node[pos=0.8, above right, font=\scriptsize, black] {$(1+|\mu(p)|)\begin{pmatrix}\cos{\frac{\theta}{2}}\\ \sin{\frac{\theta}{2}}\end{pmatrix}$};
      \draw[->, blue] (0,0) -- (0,-0.5) node[pos=0.7, below right, font=\scriptsize, black] {$(1-|\mu(p)|)\begin{pmatrix}\sin{\frac{\theta}{2}}\\-\cos{\frac{\theta}{2}}\end{pmatrix}$};
    \end{scope}
  \end{scope}

\end{tikzpicture}
    \caption{Illustration of how the Beltrami coefficient $\mu$ measures the distortion of the stretch map $S$ in a quasiconformal mapping $f$, where a small circle is mapped to an ellipse with dilation $K = \frac{1+|\mu(p)|}{1-|\mu(p)|}$.}
    \label{fig:qcmapping}
\end{figure}

Mathematically, a map $f:\mathbb{C} \rightarrow \mathbb{C}$ is quasiconformal if it satisfies Beltrami's equation  
\[
\frac{\partial f}{\partial \Bar{z}} = \mu(z) \frac{\partial f}{\partial z},
\]  
for some complex-valued function $\mu$ with $\lVert \mu \rVert_\infty < 1$. The function $\mu$, called the \emph{Beltrami coefficient}, captures the local geometric distortion of $f$. Specifically, $f$ is conformal around a small neighborhood of a point $p$ if and only if $\mu(p) = 0$. Expanding $f$ locally around $p$, we have:  
\[
f(z) = f(p) + f_{z}(p)z + f_{\Bar{z}}(p)\Bar{z},
\]  
or equivalently,  
\[
f(z) = f(p) + f_{z}(p)\big(z + \mu(p) \Bar{z}\big).
\]  
This decomposition shows that $f$ consists of a translation mapping $p$ to $f(p)$, followed by a stretch map $S(z) = z + \mu \Bar{z}$, and finally a multiplication by the conformal map $f_z(p)$. Among these, $S(z)$ introduces distortion by mapping small circles to small ellipses, with the stretching controlled by $\mu(p)$.  

To analyze $S(z)$, let $z = \rho + i\tau$ and $\mu = re^{i\theta}$, where $r = |\mu|$. Representing $S(z)$ in matrix form, we have:  
\[
\left[I + r\begin{pmatrix}
    \cos{\theta} & -\sin{\theta} \\
    \sin{\theta} & \cos{\theta}
\end{pmatrix}
\begin{pmatrix}
    1 & 0 \\
    0 & -1
\end{pmatrix}\right] 
\begin{pmatrix}
    \rho \\
    \tau
\end{pmatrix} = 
\begin{pmatrix}
    1 + r\cos{\theta} & r\sin{\theta} \\
    r\sin{\theta} & 1 - r\cos{\theta}
\end{pmatrix}
\begin{pmatrix}
    \rho \\
    \tau
\end{pmatrix}.
\]  
The eigenvalues of this matrix are $1+|\mu|$ and $1-|\mu|$, with corresponding eigenvectors $(\cos{\frac{\theta}{2}}, \sin{\frac{\theta}{2}})^T$ and $(\sin{\frac{\theta}{2}}, -\cos{\frac{\theta}{2}})^T$, which are orthogonal to each other. The distortion, or dilation, is measured by $K = \frac{1+|\mu|}{1-|\mu|}$. The Beltrami coefficient $\mu$ thus provides essential information about the properties of a quasiconformal map.

\section{Geometric representation for 3D Bijective Mappings}\label{sec:continuous}
One of the primary challenges in scientific computing related to mapping problems is developing an effective geometric representation of mappings that is easy to manipulate. In this section, we will describe our proposed geometric representation, termed {\it 3D Quasiconformality}, designed specifically for 3D bijective mappings.

\subsection{3D Quasiconformality} In this subsection, we introduce and describe our proposed {\it 3D Quasiconformality (3DQC)} for 3D bijective mappings.  Let $\mathcal{M}$ denotes the space of 3D bijective mappings between two oriented and simply connected open subsets in $\mathbb{R}^3$. Our objective is to establish a correspondence between $\mathcal{M}$ and the space of 3DQCs $\mathcal{Q}$. Specifically, we seek to define a bijection $\mathcal{R}: \mathcal{M} \to \mathcal{Q}$, where each representation in $\mathcal{Q}$ captures the geometric properties of its corresponding mapping in $\mathcal{M}$.

Consider a diffeomorphism $f: \Omega_1 \to \Omega_2$, where $\Omega_1$ and $\Omega_2$ are open, oriented, and simply connected subsets of $\mathbb{R}^3$. Let $\textbf{J}_f(p)$ represent the Jacobian matrix of $f$ at a point $p \in \Omega_1$. Since $f$ is a diffeomorphism, it is invertible, and the determinant of the Jacobian, $\det(\textbf{J}_f(p))$, is strictly positive for every $p \in \Omega_1$. This condition guarantees that $f$ preserves the orientation of the mapping.

Let $\Sigma$ be the diagonal matrix containing the singular values $a, b, \text{ and }c$ of $\textbf{J}_f(p)$, where by convention, we require that $a \geq b \geq c$. By polar decomposition, we have
\begin{align}
\begin{split}
    \textbf{J}_f(p)&=U(p)P(p)\\
    P(p)&=\sqrt{\textbf{J}_f(p)^T \textbf{J}_f(p)}=W\Sigma W^{-1}
\end{split}
    \label{polar_de_equ}
\end{align}
where $U$ is a rotation matrix, because the determinant of $\textbf{J}_f(p)$ and $P(p)$ are greater than 0, implying that $\det(U(p))=1$. $a, b \text{ and } c$ are the eigenvalues of $P(p)$ representing the dilation information, as illustrated in Figure~\ref{dilation}, and $W$ is an orthogonal matrix containing eigenvectors for dilation.  

\begin{figure}[h!]
  \centering
  \includegraphics[width=0.65\textwidth]{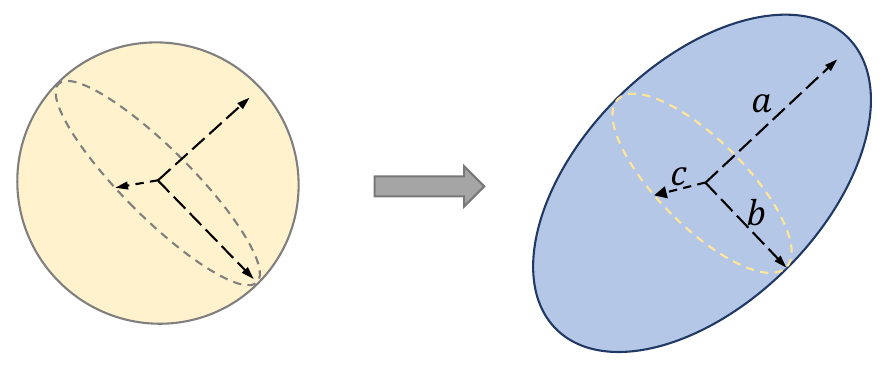}
  \caption{Illustration of the dilation of a 3D mapping. Three axes in the left diagram are mapped to those in the right diagram, maintaining their original directions with dilation factors a, b, and c. The normalized axes form the columns of matrix $W$.}
  \label{dilation}
  
\end{figure}

In 2D space, the Beltrami coefficient $\mu$ encodes the dilation factors $1+|\mu|$ and $1-|\mu|$, along with their respective directions, $(\cos{\frac{\theta}{2}}, \sin{\frac{\theta}{2}})^T$ and $(\sin{\frac{\theta}{2}}, -\cos{\frac{\theta}{2}})^T$, which are orthogonal to each other. Similarly, in 3D space, we have three dilation factors $a$, $b$, and $c$, along with their corresponding directions, represented by the column vectors of $W$. These directions are mutually orthogonal because $P(p)$ is symmetric. Thus, analogous to the Beltrami coefficient in 2D, $P(p)$ encodes the dilation information of the 3D mapping $f$.

Furthermore, the matrix \( W \) can be represented using Euler angles. To achieve this, the column vectors of \( W \) must form an orthonormal basis, and it is required that \( \det(W) = 1 \). However, since \( \det(W) \) is not necessarily positive, \( W \) may not initially be a rotation matrix. This issue can be resolved by randomly multiplying one of the columns of \( W \) by \(-1\), resulting in a modified matrix \( \tilde{W} \). This adjustment ensures that \( \det(\tilde{W}) = 1 \), making \( \tilde{W} \) a valid rotation matrix, while leaving the matrix \( P(p) \) unchanged.

To understand why $P(p)$ remains unaffected, consider that it can be expressed as:
\[
P(p) = a\vec{w}_1\vec{w}_1^T + b\vec{w}_2\vec{w}_2^T + c\vec{w}_3\vec{w}_3^T,
\]
where $\vec{w}_1$, $\vec{w}_2$, and $\vec{w}_3$ are the column vectors of $W$. Suppose we reverse the direction of one of these vectors, e.g., let $\tilde{\vec{w}}_1 = -\vec{w}_1$. Then:
\[
\tilde{\vec{w}}_1 \tilde{\vec{w}}_1^T = (-\vec{w}_1)(-\vec{w}_1)^T = \vec{w}_1\vec{w}_1^T.
\]

Thus, flipping the sign of $\vec{w}_1$ (or similarly, $\vec{w}_2$ or $\vec{w}_3$) does not alter $P(p)$. Consequently, there always exists an orthogonal matrix $\tilde{W}$ that satisfies $\det(\tilde{W}) = 1$, making it a valid rotation matrix. 

As this step is almost trivial, for brevity, we will no longer distinguish between $\tilde{W}$ and $W$, and use $W$ to denote the orthogonal rotation matrix produced by this step directly.

Now, $W$ can be expressed in the following form:
\begin{small}
\begin{equation}
    \begin{aligned}
        W &= \begin{pmatrix}
            r_{11} & r_{12} & r_{13} \\
            r_{21} & r_{22} & r_{23} \\
            r_{31} & r_{32} & r_{33}
        \end{pmatrix} = R_z(\theta_z) R_y(\theta_y) R_x(\theta_x) \\
        &= \begin{pmatrix}
            \cos{\theta_z} & -\sin{\theta_z} & 0 \\
            \sin{\theta_z} & \cos{\theta_z} & 0 \\
            0 & 0 & 1
        \end{pmatrix}
        \begin{pmatrix}
            \cos{\theta_y} & 0 & \sin{\theta_y} \\
            0 & 1 & 0 \\
            -\sin{\theta_y} & 0 & \cos{\theta_y}
        \end{pmatrix}
        \begin{pmatrix}
            1 & 0 & 0 \\
            0 & \cos{\theta_x} & -\sin{\theta_x} \\
            0 & \sin{\theta_x} & \cos{\theta_x}
        \end{pmatrix} \\
        &= \begin{pmatrix}
            \cos{\theta_z}\cos{\theta_y} & \cos{\theta_z}\sin{\theta_y}\sin{\theta_x} - \sin{\theta_z}\cos{\theta_x} & \cos{\theta_z}\sin{\theta_y}\cos{\theta_x} + \sin{\theta_z}\sin{\theta_x} \\
            \sin{\theta_z}\cos{\theta_y} & \sin{\theta_z}\sin{\theta_y}\sin{\theta_x} + \cos{\theta_z}\cos{\theta_x} & \sin{\theta_z}\sin{\theta_y}\cos{\theta_x} - \cos{\theta_z}\sin{\theta_x} \\
            -\sin{\theta_y} & \cos{\theta_y}\sin{\theta_x} & \cos{\theta_y}\cos{\theta_x}
        \end{pmatrix}.
    \end{aligned}
    \label{rotation_matrix}
\end{equation}
\end{small}

The Euler angles $\theta_x$, $\theta_y$, and $\theta_z$ can then be determined as:
\begin{equation}
    \begin{aligned}
        \theta_x &= \text{atan2}(r_{32}, r_{33}), \\
        \theta_y &= \text{atan2}(-r_{31}, \sqrt{r_{32}^2 + r_{33}^2}), \\
        \theta_z &= \text{atan2}(r_{21}, r_{11}).
    \end{aligned}
\end{equation}

The geometric interpretation of these Euler angles is illustrated in Figure~\ref{eu_ang}.

\begin{figure}[h!]
  \centering
  \includegraphics[width=0.75\textwidth]{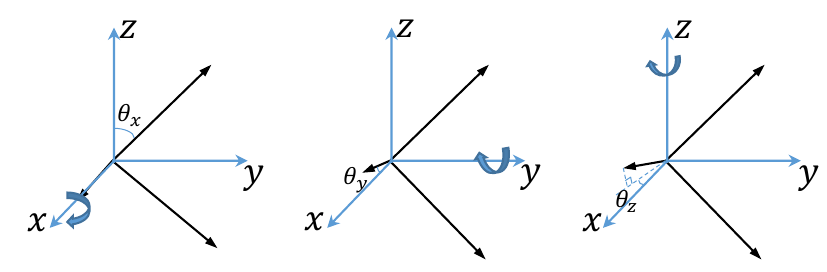}
  \caption{Illustration of the Euler angles.}
  \label{eu_ang}
  
\end{figure}

We can now formally define the 3D quasiconformality (3DQC), $\mathbf{q}$, of the bijective mapping $f$ as follows:

\begin{definition}
    Let $f: \Omega_1 \to \Omega_2$ be a bijective mapping between two oriented and simply connected open subsets $\Omega_1, \Omega_2 \subset \mathbb{R}^3$. The \textit{3D quasiconformality (3DQC)}, $\mathbf{q}:\Omega_1\to \mathbb{R}^6$, of $f$ is defined as:
    \[
    \mathbf{q}(p) = (a, b, c, \theta_x, \theta_y, \theta_z),
    \]
    where:
    \begin{itemize}
        \item $a \geq b \geq c > 0$ are the singular values of the Jacobian matrix $\textbf{J}_f(p)$, representing the principal geometric dilations of $f$ at each point $p \in \Omega_1$.
        \item $\theta_x, \theta_y, \theta_z$ are the Euler angles, representing the rotational properties of $f$.
    \end{itemize}
\end{definition}

\medskip

The \textit{3D quasiconformality (3DQC)}, $\mathbf{q}$, captures the geometric properties of the mapping $f$. Specifically:
\begin{itemize}
    \item The dilation factors $a$, $b$, and $c$ describe the local scaling of the mapping along the principal axes. See Figure~\ref{dilation}.
    \item The Euler angles $\theta_x$, $\theta_y$, and $\theta_z$ encode the direction of the principal axes. See Figure~\ref{eu_ang}. 
\end{itemize}

\subsection{Relationship between $\mathbf{q}$ and $f$}

This subsection describes the connection between the proposed 3D conformality $\mathbf{q}$ and its corresponding mapping $f$, as stated in the following theorem.

\begin{theorem}
  Let $f: \Omega_1 \to \Omega_2$ be a bijective mapping, and let $\mathbf{q}$ be its corresponding 3DQC. For every $p \in \Omega_1$, suppose $\mathbf{q} = (a, b, c, \theta_x, \theta_y, \theta_z)$, and let $f = (u, v, w)$ denote the coordinate functions of $f$. Then, we have 
  \begin{equation}
  \begin{aligned}
      \nabla \cdot \mathcal{A} \nabla u &= 0,\\
      \nabla \cdot \mathcal{A} \nabla v &= 0,\\
      \nabla \cdot \mathcal{A} \nabla w &= 0,
  \end{aligned}
  \label{PDE}
  \end{equation}
  where $\mathcal{A}$ is defined as
  \begin{equation}
      \mathcal{A} = W \begin{pmatrix}
          \frac{bc}{a} & 0 & 0\\
          0 & \frac{ac}{b} & 0\\
          0 & 0 & \frac{ab}{c}
      \end{pmatrix} W^{-1},
      \label{matrix_A}
  \end{equation}
  and $W$ is the rotation matrix determined by $\theta_x$, $\theta_y$, and $\theta_z$, as depicted in (\ref{rotation_matrix}).
\end{theorem}

\begin{proof}
According to (\ref{polar_de_equ}), we have
\begin{equation}
\textbf{J}_f(p)^T \textbf{J}_f(p) = W \Sigma \Sigma W^{-1} = W
\begin{pmatrix}
    a^2 & 0 & 0\\
    0 & b^2 & 0\\
    0 & 0 & c^2
\end{pmatrix} 
W^{-1},
\label{DTD_detQ}
\end{equation}
where $\Sigma$ is the diagonal matrix of singular values. Thus, $\det(\textbf{J}_f(p)) = \det(\Sigma) = abc$. Since $f$ is a diffeomorphism, $\textbf{J}_f(p)$ is invertible. Right-multiplying both sides of the above equation by $\textbf{J}_f(p)^{-1}$, we obtain
\begin{equation}
    \textbf{J}_f(p)^T = W 
    \begin{pmatrix}
        a^2 & 0 & 0\\
        0 & b^2 & 0\\
        0 & 0 & c^2
    \end{pmatrix}
    W^{-1} \frac{1}{\det(\textbf{J}_f(p))} C = W 
    \begin{pmatrix}
        \frac{a}{bc} & 0 & 0\\
        0 & \frac{b}{ac} & 0\\
        0 & 0 & \frac{c}{ab}
    \end{pmatrix}
    W^{-1} C,
    \label{right_mul_inv_D}
\end{equation}
where $C$ is the adjugate matrix of $\textbf{J}_f(p)$. 

For any invertible matrix $M$, we know:
\begin{equation}\label{observation_r}
    M^{-1} = \begin{pmatrix}
        \text{---}&\vec{r}_1^T&\text{---}\\
        \text{---}&\vec{r}_2^T&\text{---}\\
        \text{---}&\vec{r}_3^T&\text{---}\\
    \end{pmatrix}^{-1}=\frac{1}{\det(M)} \begin{pmatrix}
        |&|&|\\
        \vec{r}_2\times\vec{r}_3&\vec{r}_3\times\vec{r}_1&\vec{r}_1\times\vec{r}_2\\
        |&|&|
    \end{pmatrix},
\end{equation}
where $\vec{r}_1$, $\vec{r}_2$, and $\vec{r}_3$ are the rows of $M$. Hence, $C$ can be expressed as:
\begin{equation}
    C = 
    \begin{pmatrix}
        | & | & | \\
        \nabla v \times \nabla w & \nabla w \times \nabla u & \nabla u \times \nabla v \\
        | & | & |
    \end{pmatrix},
    \label{matrix_C}
\end{equation}
where $\nabla u$, $\nabla v$, and $\nabla w$ are the columns of $\textbf{J}_f(p)^T$. 

Rearranging the terms in (\ref{right_mul_inv_D}), we have:
\begin{equation}
    W 
    \begin{pmatrix}
        \frac{bc}{a} & 0 & 0\\
        0 & \frac{ac}{b} & 0\\
        0 & 0 & \frac{ab}{c}
    \end{pmatrix}
    W^{-1} \textbf{J}_f(p)^T = C.
    \label{A_gradient_C}
\end{equation}

We can rewrite (\ref{A_gradient_C}) in the following form:
\begin{equation}
    \mathcal{A} 
    \begin{pmatrix}
        | & | & | \\
        \nabla u & \nabla v & \nabla w \\
        | & | & |
    \end{pmatrix} = 
    \begin{pmatrix}
        | & | & | \\
        \nabla v \times \nabla w & \nabla w \times \nabla u & \nabla u \times \nabla v \\
        | & | & |
    \end{pmatrix}.
    \label{A_gradient_C_rewrite}
\end{equation}

For $\nabla u$, we compute:
\begin{equation}
    \nabla \cdot \mathcal{A} \nabla u = \nabla \cdot (\nabla v \times \nabla w) = \nabla w \cdot (\nabla \times \nabla v) - \nabla v \cdot (\nabla \times \nabla w) = 0,
\end{equation}
since the curl of a gradient is always zero, i.e., $\nabla \times \nabla = \mathbf{0}$. The same result holds for $\nabla v$ and $\nabla w$.

Finally, taking the divergence on both sides of (\ref{A_gradient_C_rewrite}), we obtain (\ref{PDE}), which completes the proof.
\end{proof}

\bigskip

\noindent {\bf Remark:} {\it For any point $p\in \mathcal{W} \subseteq \Omega_1$, (\ref{PDE}) can be rewritten as
\begin{equation}
\begin{aligned}
    \lim_{|\mathcal{W}|\rightarrow 0} \frac{1}{|\mathcal{W}|}\Phi(\mathcal{A}\nabla u,\partial \mathcal{W})&=0\\
    \lim_{|\mathcal{W}|\rightarrow 0} \frac{1}{|\mathcal{W}|}\Phi(\mathcal{A}\nabla v,\partial \mathcal{W})&=0\\
    \lim_{|\mathcal{W}|\rightarrow 0} \frac{1}{|\mathcal{W}|}\Phi(\mathcal{A}\nabla w,\partial \mathcal{W})&=0
\end{aligned}
\label{PDE_div_form}
\end{equation}
where $|\mathcal{W}|$ denotes the volume of $\mathcal{W}$, $\partial \mathcal{W}$ is the boundary surface of $\mathcal{W}$, $\hat{n}$ is the outward unit normal to $\partial \mathcal{W}$, and $\Phi(\textbf{F}, \partial \mathcal{W})$ is the flux of the vector field $\mathbf{F}$ across the surface $\partial \mathcal{W}$, which has the form
\begin{equation}
    \Phi(\mathbf{F},\partial \mathcal{W})=\oint_{\partial \mathcal{W}} \mathbf{F}\cdot \hat{n} dS
\end{equation}

\noindent This formulation will be useful for the implementation of the reconstruction algorithm.
}

\section{Discretisation and implementation}\label{sec:discrete}

This section describes the implementation of the proposed framework in the discretized domain. Specifically, this section explains how to compute the 3DQC from a given 3D mapping and how to reconstruct the corresponding 3D mapping from its associated 3DQC.

For the 2D case, the domain is typically triangulated to obtain finite elements, and a piecewise linear function is defined on the triangulated domain to approximate a smooth mapping. This approach generalizes naturally to 3D by defining a piecewise linear function on tetrahedra, which serve as the finite elements.

Consider a 3D simply connected domain $\Omega_1$, discretized into a tetrahedral mesh with vertices $\mathcal{V} = \{p_1, p_2, \ldots, p_n\}$ and tetrahedra represented by the index set $\mathcal{F} = \{(i_1, i_2, i_3, i_4) \mid i_1, i_2, i_3, i_4 \in \mathbb{N} \cap [1, n], \text{ and } i_1, i_2, i_3, i_4 \text{ are distinct}\}$. Additionally, the order of the vertices in a tetrahedron $T = (i_1, i_2, i_3, i_4)$ must satisfy the following orientation condition: the determinant of the matrix formed by the three vectors $(p_{i_2} - p_{i_1})$, $(p_{i_3} - p_{i_1})$, and $(p_{i_4} - p_{i_1})$ must be positive. Equivalently, this can be expressed as:
\[
[(p_{i_2} - p_{i_1}) \times (p_{i_3} - p_{i_1})] \cdot (p_{i_4} - p_{i_1}) > 0.
\]

\subsection{Compute 3DQC from a mapping}

In the above setting, the mapping $f: \Omega_1 \rightarrow \Omega_2$ is piecewise linear. To compute the 3DQC $\mathbf{q}$ for $f$, we need to compute $\mathbf{q}$ at each point in $\Omega_1$. Since $f$ is piecewise linear, the Jacobian matrices $\textbf{J}_{f}$ at all points within the same tetrahedron are equal, resulting in the same polar decomposition and 3DQC. Therefore, it suffices to compute the Jacobian matrix for each tetrahedron $T$, denoted as $\textbf{J}_{f}(T)$, and use it to determine $\mathbf{q}$. In other words, $\mathbf{q}$ is constant within each tetrahedron $T$. The linear function defined on each tetrahedron $T$ can be expressed as:
\begin{equation}
    f|_T(x, y, z) = \textbf{J}_{f}(T) \cdot \begin{pmatrix}
        x\\y\\z
    \end{pmatrix} + \begin{pmatrix}
        \vartheta_T\\\iota_T\\\varrho_T
    \end{pmatrix},
\end{equation}
where $\textbf{J}_{f}(T)$ is the Jacobian matrix of $f$ on $T$, and $\begin{pmatrix} \vartheta_T \\ \iota_T \\ \varrho_T \end{pmatrix}$ is a translation vector specific to $T$. The detailed numerical algorithm to compute the 3DQC of $f$ can be summarized as in Algorithm~\ref{alg:cap}.

\begin{algorithm}
    \caption{3DQC computation for 3D bijective map}\label{alg:cap}
    \begin{algorithmic}
        \Require $\mathcal{V} = \{p_1, p_2, \ldots, p_n\}$; $f(\mathcal{V}) = \{s_1, s_2, \ldots, s_n\}$; $\mathcal{F} = \{(i_1, i_2, i_3, i_4)\mid i_1, i_2, i_3, i_4 \in \mathbb{N}\cap [1,n], \text{ and } i_1, i_2, i_3, i_4 \text{ are distinct}\}$
        \Ensure 3DQC $\mathbf{q}(T)$ for $T \in \mathcal{F}$
        \For{$T \in \mathcal{F}$}
            \State Compute the Jacobian matrix $\textbf{J}_f(T)$ with $\mathcal{V}$, $f(\mathcal{V})$ and $T$
            \State $P \leftarrow \sqrt{\textbf{J}_f(T)^T \textbf{J}_f(T)}$
            
                \State Compute eigen decomposition: $P=WDW^T$
                \State $\Sigma \leftarrow \sqrt{D}$
                \State $a, b, c \leftarrow \Sigma_{11}, \Sigma_{22}, \Sigma_{33}$ \Comment{By convention, we assume $a\geq b \geq c$.}
                \If{$det(W)=-1$}
                    \State Update $W$ by multiplying the first column by $-1$.
                \EndIf
                \State $\theta_x \leftarrow atan2(r_{32}, r_{33})$ \Comment{$r_{ij}$ denotes the entry in the i row j column of $W$}
                \State $\theta_y \leftarrow atan2(-r_{31}, \sqrt{r_{32}^2+r_{33}^2})$
                \State $\theta_z \leftarrow atan2(r_{21}, r_{11})$
                \vspace{1em}
                \State $\mathbf{q}(T) \leftarrow (a, b, c, \theta_x, \theta_y, \theta_z)$
        \EndFor
    \end{algorithmic}
\end{algorithm}

\subsection{Reconstruction algorithm}

To reconstruct the mapping \( f \), the matrix \( \mathcal{A} \) is required, which encodes the 3D representation information. Given the representation \( \mathbf{q} \) at a point \( p \), the matrix \( \mathcal{A} \) at \( p \) can be computed according to (\ref{matrix_A}). Since every point within a tetrahedron \( T \) shares the same dilation, \( \mathcal{A} \) is a piecewise-constant matrix-valued function, where \( \mathcal{A} \) takes a constant matrix value within each tetrahedron. This constant matrix is denoted as \( \mathcal{A}_T \).

To construct the discrete version of (\ref{PDE_div_form}), the discrete gradient operator is derived. Consider a tetrahedron \( T = (i_1, i_2, i_3, i_4) \) with vertices \( p_I = (x_I, y_I, z_I) \) in \( \mathcal{V} \), and their corresponding mapped points \( s_I = f(p_I) = (u_I, v_I, w_I) \), where \( I \in \{i_1, i_2, i_3, i_4\} \). The relationship between the Jacobian matrix \( \mathbf{J}_f(T) \), the geometry of the tetrahedron, and the mapping \( f \) is captured by the equation:
\begin{equation}
    \mathbf{J}_f(T) X = Y,
\end{equation}
where
\begin{equation}
    X = \begin{pmatrix}
        | & | & | \\
        p_{i_2} - p_{i_1} & p_{i_3} - p_{i_1} & p_{i_4} - p_{i_1} \\
        | & | & |
    \end{pmatrix}, \quad
    Y = \begin{pmatrix}
        | & | & | \\
        s_{i_2} - s_{i_1} & s_{i_3} - s_{i_1} & s_{i_4} - s_{i_1} \\
        | & | & |
    \end{pmatrix}.
\end{equation}

Here, \( X \) encodes the geometric structure of the tetrahedron in \( \Omega_1 \), while \( Y \) represents its mapped counterpart in \( \Omega_2 \). Moving the matrix \( X \) to the right gives:
\begin{equation}
    \mathbf{J}_f(T) = Y X^{-1},
    \label{move_x2right}
\end{equation}
provided that \( X \) is invertible, which is guaranteed by the orientation condition.

The coefficients in \( X^{-1} \) are indexed using the notation \( \chi_{jk} \), where \( j \) and \( k \) are the row and column of the indexed entry, respectively. The coefficients are defined as:
\begin{equation}
\begin{aligned}
    A_T^{i_1} = -(\chi_{11} + \chi_{21} + \chi_{31}), \quad A_T^{i_2} = \chi_{11}, \quad A_T^{i_3} = \chi_{21}, \quad A_T^{i_4} = \chi_{31}, \\
    B_T^{i_1} = -(\chi_{12} + \chi_{22} + \chi_{32}), \quad B_T^{i_2} = \chi_{12}, \quad B_T^{i_3} = \chi_{22}, \quad B_T^{i_4} = \chi_{32}, \\
    C_T^{i_1} = -(\chi_{13} + \chi_{23} + \chi_{33}), \quad C_T^{i_2} = \chi_{13}, \quad C_T^{i_3} = \chi_{23}, \quad C_T^{i_4} = \chi_{33}.
\end{aligned}
\end{equation}
Using these, the discrete gradient operator \( \nabla_T \) is composed of:
\begin{equation}
    \partial_x = (A_T^{i_1}, A_T^{i_2}, A_T^{i_3}, A_T^{i_4})^T, \quad
    \partial_y = (B_T^{i_1}, B_T^{i_2}, B_T^{i_3}, B_T^{i_4})^T, \quad
    \partial_z = (C_T^{i_1}, C_T^{i_2}, C_T^{i_3}, C_T^{i_4})^T.
\end{equation}

The Jacobian matrix \( \mathbf{J}_f(T) \) is expressed as:
\begin{equation}
    \label{Jacobian}
    \mathbf{J}_f(T) = \begin{pmatrix}
        \partial_x \cdot u_T & \partial_y \cdot u_T & \partial_z \cdot u_T \\
        \partial_x \cdot v_T & \partial_y \cdot v_T & \partial_z \cdot v_T \\
        \partial_x \cdot w_T & \partial_y \cdot w_T & \partial_z \cdot w_T
    \end{pmatrix}.
\end{equation}

The divergence operator, defined in a connected domain \( \mathcal{W} \), is approximated in the discrete setting by replacing it with the flux computation. The flux is written as:
\begin{equation}
    \Phi_{p_i}(\mathbf{F}, \partial (\cup \mathcal{N}(p_i))) = \sum_{T \in \mathcal{N}(p_i)} \text{Area}(T_{\setminus i}) \cdot \mathbf{F}(T) \cdot \hat{n}_{T_{\setminus i}},
    \label{flux}
\end{equation}
where \( T_{\setminus i} \) denotes the triangle opposite to \( p_i \), \( \hat{n}_{T_{\setminus i}} \) is the outward unit normal to \( T_{\setminus i} \), and \( \mathbf{F}(T) \) is a piecewise-constant vector field defined on \( T \).

Substituting the derived expressions for \( \mathcal{A}_T \nabla_T u \), \( \mathcal{A}_T \nabla_T v \), and \( \mathcal{A}_T \nabla_T w \) into the flux equation gives the following system:
\begin{equation}
\begin{aligned}
    \sum_{T \in \mathcal{N}(p_i)} \text{Vol}(T) \cdot (A_T^i, B_T^i, C_T^i)^T \cdot \mathcal{A}_T \nabla_T u &= 0, \\
    \sum_{T \in \mathcal{N}(p_i)} \text{Vol}(T) \cdot (A_T^i, B_T^i, C_T^i)^T \cdot \mathcal{A}_T \nabla_T v &= 0, \\
    \sum_{T \in \mathcal{N}(p_i)} \text{Vol}(T) \cdot (A_T^i, B_T^i, C_T^i)^T \cdot \mathcal{A}_T \nabla_T w &= 0.
\end{aligned}
\label{key_linear_system}
\end{equation}

Reorganizing (\ref{key_linear_system}) leads to three linear systems:
\begin{equation}
    \mathcal{C}u = 0, \quad \mathcal{C}v = 0, \quad \mathcal{C}w = 0.
    \label{key_linear_system_mat}
\end{equation}
The matrices \( \mathcal{C} \) depend on the source mesh (\( \mathcal{V} \) and \( T \)) and the 3D quasiconformal representation (\( \mathcal{A}_T \)).

To introduce boundary conditions, a boundary vector \( \vec{v} \) for the function \( u \) is defined as:  
\[
    \vec{v} = \sum_{i=1}^{\mathcal{K}} \beta_i \vec{e}_i,
\]
where \( \vec{e}_i \) is a "one-hot" vector, and \( \beta_i \) represents the target value for the boundary vertex \( p_i \). A temporary vector \( \tilde{h} = -\mathcal{C} \vec{v} \) is computed, and its entries are updated to form \( h_u \). Similarly, \( h_v \) and \( h_w \) are computed for \( v \) and \( w \), respectively.

A "masking" operation is then applied to the matrix \( \mathcal{C} \). Rows and columns corresponding to boundary vertices are set to zero, except for their diagonal entries, which are set to 1. This operation produces the matrices \( \mathcal{C}_u \), \( \mathcal{C}_v \), and \( \mathcal{C}_w \).

The final system of equations becomes:
\begin{equation}
    \mathcal{C}_u u = h_u, \quad \mathcal{C}_v v = h_v, \quad \mathcal{C}_w w = h_w,
    \label{key_linear_system_mat_uvw}
\end{equation}
from which the mapping \( f \) is reconstructed. The conjugate gradient method is applied to solve these linear systems, as shown in section~\ref{sec:analysis}, where it is proven that the matrices \( \mathcal{C}_u \), \( \mathcal{C}_v \), and \( \mathcal{C}_w \) are symmetric positive definite.

The reconstruction algorithm is summarized in Algorithm~\ref{alg:recon}.

\begin{algorithm}
    \caption{Mapping reconstruction from 3DQC}\label{alg:recon}
    \begin{algorithmic}
        \Require $\mathcal{V} = \{p_1, p_2, \ldots, p_n\}$; $\mathcal{F} = \{(i_1, i_2, i_3, i_4)\mid i_1, i_2, i_3, i_4 \in \mathbb{N}\cap [1,n], \text{ and } i_1, i_2, i_3, i_4 \text{ are distinct}\}$;  
        $\mathbf{q}(T)$ for $T \in \mathcal{F}$. 
        \Ensure $f({\mathcal{V}})$
        \For{$T \in \mathcal{F}$}
            \State Reorganize $\vec{q}(T)$ to get the matrix $P$
            \State Compute eigenvalue decomposition of $P$ to get $a, b, c$ and $W$
            \State Compute $\mathcal{A}_T$ by (\ref{matrix_A}) with $a, b, c$ and $W$
            \State Construct $X$ and compute $X^{-1}$, from which $\{A_T^i, B_T^i, C_T^i\mid i\in T\}$ can be determined
        \EndFor
        Build the $|\mathcal{V}|\times|\mathcal{V}|$ linear systems for $\Phi_{p_i}(\mathcal{A}_T\nabla_T\textbf{f}, \partial \bigl(\cup \mathcal{N}(p_i)\bigr))=0$ for $\textbf{f}\in\{u, v, w\}$
        \State Solve the system to obtain $f({\mathcal{V}})$ using the conjugate gradient method
    \end{algorithmic}
    \label{code_recon}
\end{algorithm}

\section{Analysis}\label{sec:analysis}
This section explains why the linear systems in (\ref{key_linear_system_mat_uvw}) are symmetric positive definite (SPD). Unlike the triangular meshes used in 2D cases, tetrahedral meshes involve a much larger number of vertices, leading to significantly larger and more complex linear systems. These systems are more challenging to solve, but if they are SPD, efficient numerical methods like the conjugate gradient method can be applied to solve them effectively.

The discussion is organized into two parts: the first establishes the symmetry of $\mathcal{C}$, and the second confirms the positive definiteness of the linear systems.

\subsection{\texorpdfstring{Symmetry of $\mathcal{C}$}{Symmetry of C}}

This subsection aims to show that the matrix $\mathcal{C}$ in (\ref{key_linear_system_mat}) is symmetric. To do so, the incidence matrix of a tetrahedral mesh will firstly be examined.

The incidence matrix $E$ of a tetrahedral mesh is a matrix that describes the connectivity between the vertices and the tetrahedrons in the mesh. Specifically, it captures which vertices are part of which tetrahedrons. In our setting, the incidence matrix is an $|\mathcal{F}| \times |\mathcal{V}|$ matrix, where $|\mathcal{F}|$ is the number of tetrahedrons and $|\mathcal{V}|$ is the number of vertices. Each row corresponds to a tetrahedron, and each column corresponds to a vertex. If a vertex belongs to a tetrahedron, the corresponding entry in the matrix is set to $1$; otherwise, it is set to $0$.

Suppose a discrete function $f$ is defined on the vertices of the mesh. To sum the function values at the vertices of each tetrahedron, the computation can be performed as $Ef$, which produces the desired results. For a weighted sum, the entries can be replaced by the weights before computing $Ef$.

In this context, the goal is to construct a matrix capable of computing the discrete gradient for each tetrahedron. As explained in (\ref{Jacobian}), the gradient of $u_T$ over $x$ is computed by taking the dot product $\partial_x \cdot u_T$. In the implementation, $u$ is represented as a $|\mathcal{V}|$-dimensional vector. To compute $\partial_x \cdot u_T$ for all tetrahedrons, the first step is to construct the incidence matrix $E$ for the mesh. Replacing the values in $E$ with the corresponding $A_T^i$ produces the matrix $E_x$. The gradient over $x$ for all tetrahedrons is then computed as $E_xu$. Similarly, $E_y$ and $E_z$ can be constructed by replacing the values in $E$ with $B_T^i$ and $C_T^i$, respectively, to compute the gradients of $u_T$ over $y$ and $z$.

A $3|\mathcal{F}| \times |\mathcal{V}|$ matrix can now be constructed. For this matrix, the $(3i+1)$-th row (where $i$ ranges from $0$ to $|\mathcal{F}|-1$) is replaced with the $i$-th row of $E_x$, the $(3i+2)$-th row with the $i$-th row of $E_y$, and the $(3i+3)$-th row with the $i$-th row of $E_z$. This matrix is denoted as $\nabla$, functioning as the discrete gradient operator. The matrix $\nabla$ can be viewed as an $|\mathcal{F}|\times 1$ block matrix, where each block computes the gradient vector for the corresponding tetrahedron $T\in\mathcal{F}$. Each block contains only four nonzero columns, indexed by $i \in T$, with values $(A_T^{i}, B_T^{i}, C_T^{i})^T, i \in T$.

Next, consider the matrix $E^T$, which is $|\mathcal{V}| \times |\mathcal{F}|$. For any discrete function $g$ defined on the tetrahedrons of the mesh, $E^Tg$ sums the values over the tetrahedrons $T \in \mathcal{N}(p)$ surrounding each vertex $p$. Weighted sums can similarly be computed by replacing the values in $E^T$ with the appropriate weights.
In our case, $\nabla^T$ has a similar structure to $E^T$. Suppose we have a discrete vector field defined on the tetrahedrons of the mesh, we can put the vectors into a $3|\mathcal{F}|$-dimensional vector $\mathbf{F}$ according to the index of the tetrahedrons in $\mathcal{F}$. Then, computing $\nabla^T\mathbf{F}$ is equivalent to computing $\sum_{T\in \mathcal{N}(p_i)}(A_T^{i}, B_T^{i}, C_T^{i})^T \cdot \textbf{F}(T)$.

Furthermore, (\ref{key_linear_system}) includes a $Vol(T)$ term. A $3|\mathcal{F}| \times 3|\mathcal{F}|$ block diagonal matrix $G$ can be defined, where each diagonal block is $Vol(T) \cdot I$, a $3 \times 3$ matrix. It is evident that $G$ is symmetric positive definite (SPD).

The final component is a $3|\mathcal{F}| \times 3|\mathcal{F}|$ block diagonal matrix $\mathcal{A}$, with each diagonal block denoted as $\mathcal{A}_T$. Since the 3DQC assigned to each tetrahedron corresponds to an SPD matrix, the eigenvalues of this matrix are positive, and these eigenvalues also define the eigenvalues of $\mathcal{A}_T$. Thus, $\mathcal{A}_T$ is SPD, which ensures that $\mathcal{A}$ is SPD as well.

Combining these components, the matrix $\mathcal{C}$ is expressed as $\mathcal{C} = \nabla^T G \mathcal{A} \nabla$. Since $Vol(T) \cdot I \mathcal{A}_T$ is SPD, and $Vol(T) \cdot I$ and $\mathcal{A}_T$ are the diagonal blocks of $G$ and $\mathcal{A}$, respectively, it follows that $G \mathcal{A}$ is SPD. Consequently, $\mathcal{C}$ is a symmetric matrix. For simplicity, $G \mathcal{A}$ is denoted as $\mathcal{A}^G$ in the following section. Additionally, since $\mathcal{A}^G$ is SPD, there exists an invertible matrix $\mathcal{H}$ such that $\mathcal{A}^G = \mathcal{H}^T \mathcal{H}$.
\subsection{Positive definiteness}
In this section, we show that $\mathcal{C}_u$, $\mathcal{C}_v$, and $\mathcal{C}_w$ in the linear systems (\ref{key_linear_system_mat_uvw}) are SPD. For simplicity, the number of vertices is denoted by $N$ in this section (i.e., $N = |\mathcal{V}|$).

\begin{lemma}\label{rank_nabla_a_nabla}
$rank(\mathcal{C})=N-1$, and the vector $\Vec{1}=(1, 1, \ldots, 1)^T$ is the only basis vector in the null space of $\mathcal{C}$.
\end{lemma}

\begin{proof}
    Matrix $\nabla$ computes the gradient of a discrete function $\vec{\alpha}$ on tetrahedrons. If $\nabla \vec{\alpha} = 0$ for some vector $\vec{\alpha}$, then by definition, the gradient of $\vec{\alpha}$ on every tetrahedron is zero, meaning that $\vec{\alpha} = \gamma\vec{1}, \Vec{1} = (1, 1, \ldots, 1)^T$ for some constant $\gamma$ on each tetrahedron. Therefore, $\Vec{1}$ is the only basis vector in the null space of $\nabla$. This shows that $\Vec{1}$ is in the null space of $\mathcal{C}$.
    
    Then we aim to prove the uniqueness. If there exists a nonzero vector $\vec{\zeta}$ such that $\nabla^T \mathcal{H}^T \vec{\zeta} = 0$, then $\vec{\zeta}$ is orthogonal to the row vectors of $\nabla^T \mathcal{H}^T$, or equivalently, $\vec{\zeta} \notin \textbf{C}(\mathcal{H} \nabla)$, where $\textbf{C}$ denotes the columns space. However, for any nonzero vector $\vec{\alpha}$ satisfying $\vec{\alpha}\neq \gamma\vec{1}$ for any constant $\gamma$, we have nonzero vector $\vec{\delta} = \mathcal{H} \nabla \vec{\alpha}$, $\vec{\delta} \in \textbf{C}(\mathcal{H} \nabla)$, indicating that $\nabla^T \mathcal{H}^T \vec{\delta} \neq 0$. In other words, for any nonzero vector $\vec{\alpha}\neq \gamma\vec{1}$, there does not exist a nonzero vector $\vec{\delta} = \mathcal{H} \nabla \vec{\alpha}$ such that $\nabla^T \mathcal{H}^T \vec{\delta} = 0$. Therefore, the null space of $\mathcal{C}$ is only span by $\Vec{1}$.
    
    Hence, $ \text{rank}(\mathcal{C}) = N-1 $.
\end{proof}

\begin{lemma}\label{linear_dependent}
Any column (or row) vector in $\mathcal{C}$ can be expressed as a linear combination of the rest $N-1$ column (or row) vectors. 
\end{lemma}

\begin{proof}
    Matrix $\mathcal{C}$ is symmetric and therefore diagonalizable. Let $\mathcal{C} = Q \Lambda Q^T $, where $ Q $ is an orthogonal matrix containing the eigenvectors as columns, and $ \Lambda $ is a diagonal matrix with the eigenvalues. By Lemma~\ref{rank_nabla_a_nabla}, $\text{rank}(\mathcal{C}) = N-1$, indicating that one of its eigenvalues is zero, and the rest are non-zero. Without loss of generality, let the last diagonal entry in $ \Lambda $ be 0, and the last column of $ Q $ be the vector $ \frac{\Vec{1}}{\|\Vec{1}\|}=\frac{1}{\sqrt{N}}\Vec{1} $. In this setting, the last row of $ Q^T $ is equal to $\frac{1}{\sqrt{N}}\Vec{1}^T $.

    Each column of $\mathcal{C}$ can be expressed as a linear combination of the columns in $ Q \Lambda $, or equivalently, the first $ N-1 $ columns in $ Q\Lambda $, as the last eigenvalue is $0$. The linear combination coefficients are stored in $ Q^T $.

    Let $ \tilde{Q} $ be the first $ N-1 $ columns of $ Q $, then $ \tilde{Q}^T $ is the first $ N-1 $ rows of $ Q^T $. Let $\varepsilon$ be the entry in the last row of $ Q^T $, then $\varepsilon = \frac{1}{\sqrt{N}} $. Denote the column vectors in $\tilde{Q}^T$ as $\vec{\gamma}_i$ for $ i = 1, 2, \ldots, N $, then since $Q^T$ is an orthogonal matrix, the column vectors are perpendicular to each other, and the length of each vector is $1$, we have:
    \begin{align*}
        \vec{\gamma}_i \cdot \vec{\gamma}_i &= 1-\varepsilon^2 = \frac{N-1}{N}=-\frac{1}{N}(1-N)\\
        \vec{\gamma}_i \cdot \vec{\gamma}_j &= 0-\varepsilon^2 = -\frac{1}{N} \quad \text{for} \quad i \neq j
    \end{align*}
    
    Randomly pick $ N-1 $ column vectors from $ \tilde{Q}^T $ to form the matrix $\Gamma$. Then we have:
    \begin{equation*}
        \Psi = \Gamma^T \Gamma = -\frac{1}{N} \begin{pmatrix}
        1-N & 1 & \ldots & 1 \\
        1 & 1-N & \ldots & 1 \\
        \vdots & \vdots & \vdots & \vdots \\
        1 & 1 & \ldots & 1-N
        \end{pmatrix} = -\frac{1}{N}(\mathcal{J} - N I)
    \end{equation*}
    where $ I $ is an $ (N-1) \times (N-1) $ identity matrix, and $ \mathcal{J} $ is a matrix full of 1s.

    Let $\tilde{\Psi} = \mathcal{J} - N I$. Suppose $\lambda$ is an eigenvalue of $\tilde{\Psi}$, then $\det(\tilde{\Psi} - \lambda I) = 0$, or $\det(\mathcal{J} - (N + \lambda) I) = 0$. The eigenvalue of $\mathcal{J}$ is $N-1$ on the eigenvector $\Vec{1}$, and 0 on the other eigenvectors. Therefore, the eigenvalues of $\mathcal{J} - (N + \lambda) I$ are $-1 - \lambda $ on $ \Vec{1}$, and $\lambda + N$ on the other eigenvectors. Hence:
    \begin{equation*}
        \det(\tilde{\Psi} - \lambda I) = (-1 - \lambda) \cdot (\lambda + N)^{N-2} = 0
    \end{equation*}
    
    Thus, $\lambda = -1$ with multiplicity $1$, or $\lambda = -N$ with multiplicity $N-2$, and:
    \begin{equation*}
        \det(\tilde{\Psi}) = -(-N)^{N-2}
    \end{equation*}
    which implies:
    \begin{equation*}
        \det(\Psi) = \det(-\frac{1}{N}\tilde{\Psi})= \left( -\frac{1}{N} \right)^{N-1} \cdot (-(-N)^{N-2}) = \frac{1}{N} > 0
    \end{equation*} 
    which indicates $\Psi$ is full rank, and $\Gamma$ is also full rank. Therefore, any $N-1$ column vectors in $ \tilde{Q}^T $ form a basis in the $(N-1)$-dimensional space, and the remaining vector can be expressed as a linear combination of the basis.

    Let $Q_{\Lambda}$ be the first $N-1$ columns of $Q\Lambda$ and let $\tilde{Q}_{\setminus i}^T = (\vec{\gamma}_1, \vec{\gamma}_2, \ldots, \vec{\gamma}_{i-1}, \vec{\gamma}_{i+1}, \ldots, \vec{\gamma}_{N})$. Then by the above proof, we know that there exists an $N-1$ dimensional nonzero vector $\Vec{\alpha}$ such that $\vec{\gamma}_i = \tilde{Q}_{\setminus i}^T \Vec{\alpha}$. Then we have $Q_{\Lambda} \vec{\gamma}_i = Q_{\Lambda} \tilde{Q}_{\setminus i}^T \Vec{\alpha}$. 
    
    In the left-hand side, $Q_{\Lambda} \vec{\gamma}_i$ is the $i$-th column of $\mathcal{C}$, while in the right-hand side, $Q_{\Lambda} \tilde{Q}_{\setminus i}^T$ is the matrix containing the remaining $N-1$ column vectors of $\mathcal{C}$. The equation $Q_{\Lambda} \vec{\gamma}_i = Q_{\Lambda} \tilde{Q}_{\setminus i}^T \Vec{\alpha}$ means that the $i$-th column vector of $\mathcal{C}$ can be expressed as a linear combination of the remaining $N-1$ column vectors of $\mathcal{C}$, and the coefficients are stored in $\Vec{\alpha}$. 
    
    As the index $i$ can be randomly chosen, we prove the statement for column vectors. Since the matrix is symmetric, this property can be extended to the row vectors as well.
\end{proof}

The $i$-th row and $i$-th column of $\mathcal{C}$ can be set to zero. From Lemma~\ref{linear_dependent}, it is known that the rank of this matrix remains $N-1$. The vector $\Vec{e}_i$ lies in the null space of this matrix. By setting the $i$-th diagonal entry to $1$, the matrix becomes full rank and is denoted as $\mathcal{M}$. 

\begin{proposition}\label{first_masking_SPD}
    $\mathcal{M}$ is an SPD matrix.
\end{proposition}

\begin{proof}
    Setting the $i$-th column in $\nabla$ to zero, we get a new matrix denoted as $\tilde{\nabla}$. By the above definition, we have
    \begin{equation*}
        \mathcal{M}=\tilde{\nabla}^T\mathcal{A}^G\tilde{\nabla}+\mathcal{E}
    \end{equation*}
    where $\mathcal{E}$ is an $N \times N$ matrix whose $i$-th diagonal entry is $1$, and all other entries are zero. The matrix $\tilde{\nabla}^T\mathcal{A}^G\tilde{\nabla}$ is the same matrix obtained by setting $i$-th row and $i$-th column of $\mathcal{C}$ to $0$, the rank of which is $N-1$. Then we decompose any nonzero vector $\vec{v}$ into $\vec{v}_{\setminus i}$ and $\eta \Vec{e}_i$, where $\eta$ is the value in the $i$-th entry of $\vec{v}$, and $\vec{v}_{\setminus i}$ denotes the vector with the $i$-th component removed.

    Then, we have
    \begin{align*}
        \vec{v}^T\mathcal{M}\vec{v}&=(\vec{v}_{\setminus i}+\eta \Vec{e}_i)^T(\tilde{\nabla}^T\mathcal{A}^G\tilde{\nabla}+\mathcal{E})(\vec{v}_{\setminus i}+\eta \Vec{e}_i)\\
        &=\vec{v}_{\setminus i}^T\tilde{\nabla}^T\mathcal{A}^G\tilde{\nabla}\vec{v}_{\setminus i}+\eta^2 \Vec{e}_i^T\mathcal{E}\Vec{e}_i\\
        &=\|\mathcal{H}\tilde{\nabla}\vec{v}_{\setminus i}\|_2^2+\eta^2
    \end{align*}

    Since $rank(\tilde{\nabla}^T\mathcal{A}^G\tilde{\nabla})=N-1$ and $\mathcal{A}^G$ is full rank, $rank(\tilde{\nabla})=N-1$. Together with the definition of $\tilde{\nabla}$, we know that its null space only includes $\Vec{e}_i$ and its row space contains the remaining components. Since $\vec{v}_{\setminus i}$ contains the complementary components of $\Vec{e}_i$, $\tilde{\nabla}\vec{v}_{\setminus i}$ is nonzero if $\vec{v}_{\setminus i}$ is nonzero. Otherwise, if $\vec{v}_{\setminus i}$ is a zero vector, then $\eta$ is nonzero because $\vec{v}$ is nonzero. Therefore, the two terms $\|\mathcal{H}\tilde{\nabla}\vec{v}_{\setminus i}\|_2^2$ and $\eta^2$ cannot be zero at the same time. Then $\vec{v}^T\mathcal{M}\vec{v}>0$, which proves the statement.
\end{proof}

For any SPD matrix $\mathcal{M}$, setting the $i$-th row and $i$-th column to zero and the $i$-th diagonal entry to $1$ produces a new matrix, denoted as $\mathcal{P}$. 
\begin{proposition}\label{remaining_op_SPD}
    $\mathcal{P}$ is also SPD.
\end{proposition}
\begin{proof}
    Since $\mathcal{M}$ is an SPD matrix, it can be decomposed to $Q\Lambda Q^T$ where $\Lambda$ contains the eigenvalues, which are all positive. 

    Then, by setting the $i$-th row of $Q$ to $0$, we get another matrix $\tilde{Q}$. Then we have
    \begin{equation*}
        \mathcal{P}=\tilde{Q}\Lambda \tilde{Q}^T+\mathcal{E}
    \end{equation*}
    where $\mathcal{E}$ is an $N \times N$ matrix whose $i$-th diagonal entry is $1$, and all other entries are zero.

    Then we decompose any nonzero vector $\vec{v}$ into $\vec{v}_{\setminus i}$ and $\eta \Vec{e}_i$, where $\eta$ is the value in the $i$-th entry of $\vec{v}$, and $\vec{v}_{\setminus i}$ denotes the vector with the $i$-th component removed.

    Then we have
    \begin{align*}
        \vec{v}^T\mathcal{P}\vec{v}&=(\vec{v}_{\setminus i}+\eta \Vec{e}_i)^T(\tilde{Q}\Lambda \tilde{Q}^T+\mathcal{E})(\vec{v}_{\setminus i}+\eta \Vec{e}_i)\\
        &=\vec{v}_{\setminus i}^T\tilde{Q}\Lambda \tilde{Q}^T\vec{v}_{\setminus i}+\eta^2 \Vec{e}_i^T\mathcal{E}\Vec{e}_i\\
        &=\|\Sigma \tilde{Q}^T \vec{v}_{\setminus i}\|_2^2+\eta^2
    \end{align*}
    where $\Sigma=\sqrt{\Lambda}$.
    Since $rank(Q)=N$, $rank(\tilde{Q})=N-1$. Apparently, we notice that the vector $\Vec{e}_i$ is the only basis in the null space of $\tilde{Q}^T$. Since $\vec{v}_{\setminus i}$ contains the complementary components of $\Vec{e}_i$, $\tilde{Q}^T\vec{v}_{\setminus i}$ is nonzero if $\vec{v}_{\setminus i}$ is nonzero. Otherwise, if $\vec{v}_{\setminus i}$ is a zero vector, then $\eta$ is nonzero because $\vec{v}$ is nonzero. Therefore, the two terms $\|\Sigma\tilde{Q}^T\vec{v}_{\setminus i}\|_2^2$ and $\eta^2$ cannot be zero at the same time. Then $\vec{v}^T\mathcal{P}\vec{v}>0$, which proves the statement.
\end{proof}

With the above Lemmas and Propositions, we can now show that
\begin{theorem}
    $\mathcal{C}_u$, $\mathcal{C}_v$, and $\mathcal{C}_w$ in the linear systems (\ref{key_linear_system_mat_uvw}) are SPD.
\end{theorem}

\begin{proof}
    In Lemma~\ref{rank_nabla_a_nabla}, $rank(\mathcal{C})=N-1$ is proved. 

    For the first masking operation, by Lemma~\ref{linear_dependent} and Proposition~\ref{first_masking_SPD}, we know that the resulting matrix is SPD. 
    
    For the remaining masking operations, by Proposition\ref{remaining_op_SPD}, we know that the resulting matrices are all SPD. Therefore, $\mathcal{C}_u$, $\mathcal{C}_v$, and $\mathcal{C}_w$ in the linear systems (\ref{key_linear_system_mat_uvw}) are SPD.
\end{proof}

\section{Applications of 3DQC in 3D Mapping Problems}\label{sec:app} Leveraging 3DQC, we propose two algorithms for the numerical interpolation and sparse modeling of 3D bijective mappings. In this section, we provide a detailed explanation of the two proposed algorithms.

\subsection{Numerical interpolation of 3D Bijective Mappings}
Numerical interpolation of 3D bijective mappings has numerous applications, such as 3D animation, dynamic mesh generation for moving interfaces, dynamic domain generation for partial differential equations (PDEs), and deformation analysis. Mathematically, given two 3D bijective mappings $f_1:\Omega_1 \to \Omega_2$ and $f_2:\Omega_1 \to \Omega_2$, the goal is to construct a homotopy $\mathcal{H}:\Omega_1 \times [0,1] \to \Omega_2$ such that $\mathcal{H}(\cdot, 0) = f_1$, $\mathcal{H}(\cdot, 1) = f_2$, and $\mathcal{H}(\cdot, t)$ remains bijective for all $t \in [0,1]$. The main challenge is ensuring the bijectivity of $\mathcal{H}(\cdot, t)$, as simply interpolating between the coordinate functions of $f_1$ and $f_2$ often results in the loss of this property. With the 3DQC, such a homotopy can be conveniently computed within the space of 3D bijective mappings $\mathcal{M}$.

Let $\mathbf{q}_1$ and $\mathbf{q}_2$ be the 3DQCs corresponding to the mappings $f_1$ and $f_2$, respectively. Both $\mathbf{q}_1$ and $\mathbf{q}_2$ are defined on each tetrahedron $T$ of the mesh discretizing $\Omega_1$. Specifically, each $\mathbf{q}_1(T) = (a^1(T), b^1(T), c^1(T), \theta^1_x(T), \theta^1_y(T), \theta^1_z(T))$ gives rise to a symmetric positive definite matrix $P_1(T)$. More precisely, 
\[
P_1(T) = W_1(T) \, \textbf{diag}([a^1(T), b^1(T), c^1(T)]) \, W_1(T)^{-1},
\]
\noindent where 
\begin{itemize}
    \item $\textbf{diag}([a^1(T), b^1(T), c^1(T)])$ is a diagonal matrix whose diagonal entries are given by $a^1(T), b^1(T), c^1(T)$ and;
    \item $W_1(T)  = R_z(\theta^1_z(T)) R_y(\theta^1_y(T)) R_x(\theta^1_x(T))$.
\end{itemize}

\medskip

We then define a map $P_1: \mathcal{F} \to S^3_{++}$, where: $\mathcal{F}$ denotes the collection of tetrahedrons in the mesh discretizing $\Omega_1$, and $S^3_{++}$ is the space of $3 \times 3$ symmetric positive definite matrices. $P_1$ describes the local geometric properties of the mapping $f_1$. Similarly, $\mathbf{q}_2(T) = (a^2(T), b^2(T), c^2(T), \theta^2_x(T), \theta^2_y(T), \theta^2_z(T))$ gives rise to a symmetric positive definite matrix $P_2(T)$, where
\[
P_2(T) = W_2(T) \, \textbf{diag}([a^2(T), b^2(T), c^2(T)]) \, W_2(T)^{-1},
\]
\noindent where $\textbf{diag}([a^2(T), b^2(T), c^2(T)])$ and $W_2(T)$ are defined similarly as:

\medskip

\begin{itemize}
\item $\textbf{diag}([a^2(T), b^2(T), c^2(T)])$ is a diagonal matrix whose diagonal entries are given by $a^2(T), b^2(T), c^2(T)$
    \item $W_2(T) = R_z(\theta^2_z(T)) R_y(\theta^2_y(T)) R_x(\theta^2_x(T))$.
\end{itemize}

\medskip

\noindent As in the case of $P_1$, the matrix $P_2$ describes the local geometric properties of the mapping $f_2$.

To compute a homotopy between $f_1$ and $f_2$, we employ a strategy that interpolates between the SPD matrices $P_1$ and $P_2$, rather than directly interpolating between the coordinate functions of $f_1$ and $f_2$. Specifically, we aim to construct a continuous map $\tilde{\mathcal{H}} : \Omega_1\times [0,1] \to M_{3 \times 3}(\mathbb{R})$ such that $\tilde{\mathcal{H}}(p,0) = P_1(p)$ and $\tilde{\mathcal{H}}(p,1) = P_2(p)$ for all $p\in \Omega_1$. Each intermediate matrix, denoted as $\tilde{\mathcal{H}}(\cdot, t)$, corresponds to a 3DQC $\mathbf{q}_t$. Using $\mathbf{q}_t$ and the prescribed boundary conditions, a 3D mapping $f_t$ is computed via Algorithm~\ref{alg:recon}. This mapping, $f_t$, provides a smooth interpolation between $f_1$ and $f_2$. To control the bijectivity of $f_t$, we enforce that matrix $\tilde{\mathcal{H}}(p,t)$ remains symmetric positive definite for all $p\in \Omega_1$ and $t \in [0,1]$.

To construct $\tilde{\mathcal{H}}$, we compute the geodesic between $P_1(p)$ and $P_2(p)$ using the log-Euclidean metric for every $p\in \Omega_1$. Specifically, the log-Euclidean metric defines the geodesic path $\tilde{\mathcal{H}}(p,t)$ between $P_1(p)$ and $P_2(p)$ as follows:  
\begin{equation}
\tilde{\mathcal{H}}(p,t) = \exp\big((1-t)\log(P_1(p)) + t\log(P_2(p))\big), \quad t \in [0,1],
\end{equation}
where $\log(\cdot)$ and $\exp(\cdot)$ denote the matrix logarithm and matrix exponential, respectively. 

One advantage of using the log-Euclidean metric is that it avoids the swelling effect, a phenomenon where interpolating SPD matrices produces intermediate matrices with unnecessarily large eigenvalues \cite{feragen2017geometries}. This inflation leads to distorted or ``swollen" results. Such swelling effects are common when employing metrics like the affine-invariant Riemannian metric, where interpolations may disproportionately emphasize certain directions. On the other hand, the log-Euclidean metric ensures that the eigenvalues of the intermediate SPD matrices evolve smoothly. As such, their corresponding 3D bijective mapping can evolve naturally without any artificial inflation. Furthermore, the log-Euclidean metric is computationally simple. This allows for the development of a numerical interpolation algorithm for 3D bijective mapppings, ensuring both efficiency and accuracy.

In summary, the numerical interpolation algorithm for 3D bijective mapping can be summarized as follows:

\begin{algorithm}
    \caption{Interpolation of 3D bijective mappings}
    \begin{algorithmic}
        \Require $f_1:\Omega_1 \to \Omega_2$; $f_2:\Omega_1 \to \Omega_2$; $t \in [0,1]$; Boundary conditions.
        \Ensure $f_t:\Omega_1 \to \Omega_2$
        \For{$T \in \Omega_1$}
            \State Compute $\mathbf{q}_1(T)$ and $\mathbf{q}_2(T)$ from $f_1$ and $f_2$ using Algorithm~\ref{alg:cap}
            \State Compute $P_1(T)$ and $P_2(T)$ using $\mathbf{q}_1(T)$ and $\mathbf{q}_2(T)$
            \State $P_t(T)=\exp\big((1-t)\log(P_1(T)) + t\log(P_2(T))\big)$
            \State Get the corresponding $\mathbf{q}_t(T)$ from $P_t(T)$
        \EndFor
        \State Compute $f_t$ using Algorithm~\ref{alg:recon} with $\mathbf{q}_t$ and the prescribed boundary conditions.
    \end{algorithmic}
    \label{code_interp3}
\end{algorithm}
\subsection{Sparse modelling of 3D Bijective Mappings}
Compared to 2D mapping, the storage requirements for 3D mappings are significantly higher. The rapid developments in deep learning have further increased the need to create and store large datasets of 3D mappings. This necessitates the development of methods for the sparse representation of 3D bijective mappings. In this subsection, we introduce our proposed algorithm for sparse modeling of 3D bijective mappings using 3DQC.

Our idea is similar to that in \cite{lui2013texture}, in which the author compresses 2D mapping using discrete Fourier transform. However, when it comes to 3D mapping, the number of vertices increases dramatically. When the distribution of the vertices is very uneven, one may need a very dense regular grid to preserve the information of the original mesh precisely, which will boost the computation. 

As we already have a mesh structure, it is a natural way to perform spectral analysis on the tetrahedral mesh directly. \cite{shuman2013emerging} introduces the idea of performing signal processing on a graph utilizing the graph Laplacian. In a discrete 2D surface with geometric information, signal processing can be applied with the Laplace-Beltrami operator, as introduced in \cite{rustamov2007laplace,crane2018discrete}. In the 3D case, the Laplace-Beltrami operator can also be computed using a 3D version cotangent formula \cite{crane2019n, Volumetric_david}. 

The 3D discrete Laplace-Beltrami operator is a matrix $L \in \mathbb{R}^{|\mathcal{V}|\times|\mathcal{V}|}$. Let $\Upsilon$ be the subset of $\mathcal{F}$ containing the tetrahedrons incident to an edge connecting $p_i$ and $p_j$, then the weight of this edge is computed as 
\begin{equation*}
    \omega_{ij}=\frac{1}{6}\sum_{T\in \Upsilon}\ell_{T_{\setminus ij}} \cot \theta^{ij}_{T_{\setminus ij}}
\end{equation*}
where $T_{\setminus ij}$ is the edge connecting the other two vertices in $T\setminus \{i,j\}$, $\ell_{T_{\setminus ij}}$ is the length of $T_{\setminus ij}$; $\theta^{ij}_{T_{\setminus ij}}$ is the interior dihedral angle of edge connecting $p_i$ and $p_j$.

Then the nonzero entries representing the edge connecting vertex $p_i$ and $p_j$ is $$L_{ij}=-\omega_{ij}$$
and diagonal element $$L_{ii} = \sum_{j=1}^{|\mathcal{V}|}\omega_{ij}$$ for each vertex $p_i$.

Then, by solving the generalized eigenvalue problem 
\begin{equation} \label{generalize_eigenvalue_problem}
    L\nu = \lambda M \nu
\end{equation}
where $M$ is the mass matrix whose diagonal $M_{ii}=\frac{1}{4}\sum_{T\in\mathcal{N}(p_i)}Vol(T)$, we obtain a series of eigenvalues $\{\lambda_1, \lambda_2, \dots ,\lambda_{|\mathcal{V}|}\}$ satisfying $\lambda_1 < \lambda_2 < \dots < \lambda_{|\mathcal{V}|}$, and their corresponding eigenvector $\{\nu_1, \nu_2, \dots, \nu_{|\mathcal{V}|}\}$ satisfying $\nu_i^TM\nu_i=1$ for all generalized eigenvectors $\nu_i$. 

Let $\tilde{\nu}_i=M^{\frac{1}{2}}\nu_i$, then $\nu_i=M^{-\frac{1}{2}}\tilde{\nu}_i$. Substituting it into (\ref{generalize_eigenvalue_problem}), we have $LM^{-\frac{1}{2}}\tilde{\nu}_i=\lambda M^{\frac{1}{2}}\tilde{\nu}_i$, which is equivalent to $M^{-\frac{1}{2}}LM^{-\frac{1}{2}}\tilde{\nu}_i=\lambda \tilde{\nu}_i$. Note that $M^{-\frac{1}{2}}LM^{-\frac{1}{2}}$ is a symmetric matrix, and $\tilde{\nu}_i$ is its eigenvector, meaning that $\tilde{\nu}_j^T \tilde{\nu}_i=0$ if $i\neq j$. Thus, $\nu_j^T M^{\frac{1}{2}}  M^{\frac{1}{2}}\nu_i=\nu_i^T M \nu_j=0$, meaning that the generalized eigenvectors of $L$ are $M$-orthogonal to each other. We have $V^TMV=I$ for matrix $V$ whose column vectors are the generalized eigenvectors. Thus we have $(V^TM)^{-1}=V$.

For any discrete function $h$ defined on the vertices of a tetrahedral mesh, the $M$-inner product between $h$ and eigenvector $\nu_i$ is defined as  
$$\langle h, \nu_i\rangle_{M} = h^T M \nu_i.$$

Then,  
$$\sum_i \langle h, \nu_i\rangle_{M} \nu_i = \sum_i (h^T M \nu_i)\nu_i = V(V^T M h) = V(V^T M)h = h,$$  
as $(V^T M)^{-1} = V.$

Therefore, we can project any function defined on a tetrahedral mesh to the spectral domain by computing the $M$-inner product with the generalized eigenvectors of the discrete Laplace-Beltrami operator. Also, the eigenvalues indicate the frequency of the eigenvectors: a larger eigenvalue corresponds to a higher frequency. With this property, we can perform low-pass or high-pass filtering on any function defined on the mesh.

In our case, suppose we have a pair of 3D meshes with the same number of vertices and the same triangulation. Together, these meshes form a mapping. Our goal is to find a sparse representation of its 3DQC in order to compress the mapping representation.

We start by constructing the discrete Laplace-Beltrami operator and the mass matrix for the source mesh. Then, we solve the generalized eigenvalue problem (\ref{generalize_eigenvalue_problem}) with the constraint $\nu_i^T M \nu_i = 1$ for all generalized eigenvectors $\nu_i$. This gives us the sorted eigenvalues and eigenvectors.

With the eigenvalues and eigenvectors in hand, we can begin filtering the 3DQC of the mapping. Recall that the 3DQC is a piecewise linear, vector-valued function defined on the tetrahedrons. We can treat this function as a combination of six scalar-valued functions and perform filtering on each function separately.

For each of the six functions, we first interpolate it to the vertices of the mesh. Then, for each eigenvector, we compute the $M$-inner product $\xi_i = \langle f, \nu_i \rangle_{M}$ to obtain the spectrum. Based on the spectrum, we can select a threshold $\mathbf{T} \leq |\mathcal{V}|$, saving only $\{\xi_i \mid i \in \mathbb{N} \cap [1, \mathbf{T}]\}$ and discarding the rest. A smaller $\mathbf{T}$ results in greater storage savings but reduces reconstruction accuracy.

To reconstruct the function, we retrieve the saved spectral components $\{\xi_i \mid i \in \mathbb{N} \cap [1, \mathbf{T}]\}$ from storage. Then, by computing  
\begin{equation}\label{decompress_equ}
\tilde{h} = \sum_{i=1}^{\mathbf{T}} \xi_i \nu_i,
\end{equation}  
we can reconstruct the vector-valued function representing the 3DQC. Subsequently, the mapping can be recovered by solving the linear systems (\ref{key_linear_system}) using the reconstructed 3DQC.

As a remark, one might question why we do not apply the scheme to the coordinate function directly. Our experiments, as shown in the following section, reveal that this idea leads to poor results that cannot preserve the bijectivity.

In summary, the sparse modelling algorithm for 3D bijective mapping can be summarized as Algorithm~\ref{code_compress}.

\begin{algorithm}[h!]
    \caption{Sparse modeling}
    \begin{algorithmic}
        \Require $\mathcal{V} = \{p_1, p_2, \ldots, p_n\}$; $\mathcal{F} = \{(i_1, i_2, i_3, i_4)\mid i_1, i_2, i_3, i_4 \in \mathbb{N}\cap [1,n], \text{ and } i_1, i_2, i_3, i_4 \text{ are distinct}\}$;$f(\mathcal{V}) = \{s_1, s_2, \ldots, s_n\}$; $\mathbf{T} \in \mathbb{N}\cap \left[0,|\mathcal{V}|\right]$.
        \Ensure spectral components $\{\xi_i\mid i \in \mathbb{N}\cap [1, \mathbf{T}]\}$.
        \State Compute the discrete Laplace-Beltrami operator $L$ and the mass matrix $M$
        \State Solve (\ref{generalize_eigenvalue_problem}) to get eigenvalues $\{\lambda_1, \lambda_2, \dots ,\lambda_{|\mathcal{V}|}\}$ and eigenvectors $\{\nu_1, \nu_2, \dots, \nu_{|\mathcal{V}|}\}$
        \State Compute the 3DQC $\mathbf{q}$ of $f$ using Algorithm~\ref{alg:cap}
        \For{each component $\tilde{\mathbf{q}} \in \mathbf{q}$}
            \For{$i \in \mathbb{N}\cap [1, \mathbf{T}]$}
                \State $\xi_i =\langle \tilde{\mathbf{q}}, \nu_i\rangle_{M} $
            \EndFor
        \EndFor 
    \end{algorithmic}
    \label{code_compress}
\end{algorithm}


\section{Experiments}\label{sec:experiments} 
Extensive experiments have been carried out to test the effectiveness of our proposed framework. In this section, experimental results will be presented. We first test the reconstruction algorithm to reconstruct the 3D bijective mapping from its corresponding 3DQC. We subsequently present experiments for the applications mentioned above.

\subsection{Demonstration for the reconstruction algorithm}

\afterpage{
\clearpage
\begin{figure}[t]
    \centering
    \begin{subfigure}[t]{0.32\textwidth}
        \centering
	    \includegraphics[width=\textwidth]{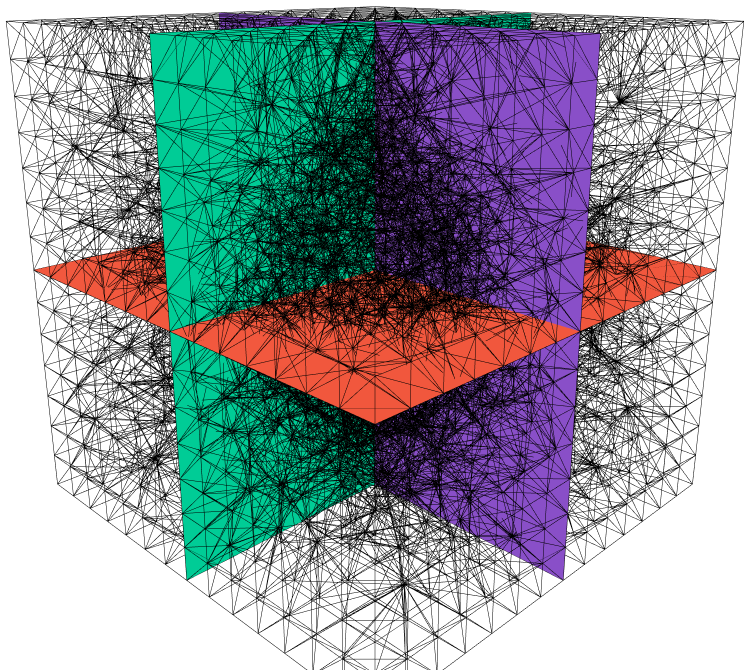}
        \caption{The source mesh.}
    \end{subfigure}
    \begin{subfigure}[t]{0.32\textwidth}
        \centering
        \includegraphics[width=\textwidth]{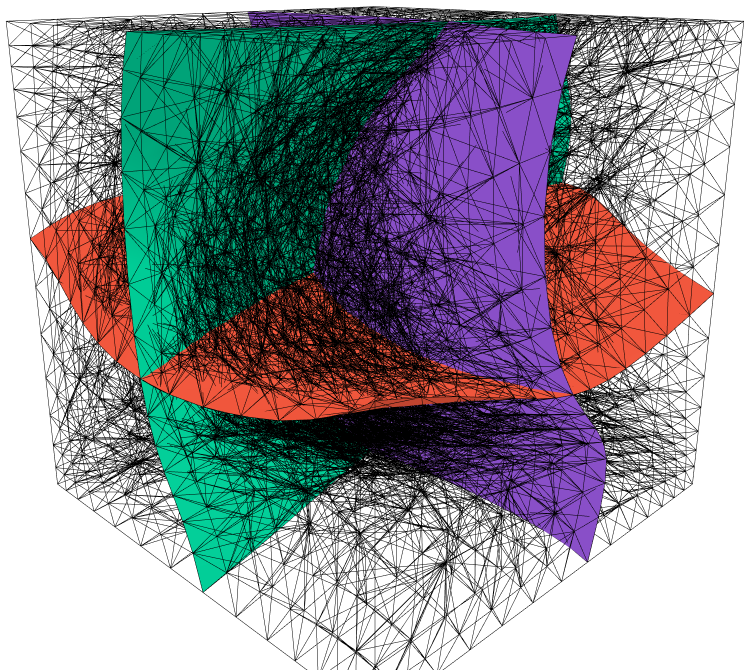} 
        \caption{The target mesh.}
    \end{subfigure}
    \begin{subfigure}[t]{0.32\textwidth}
        \centering
        \includegraphics[width=\textwidth]{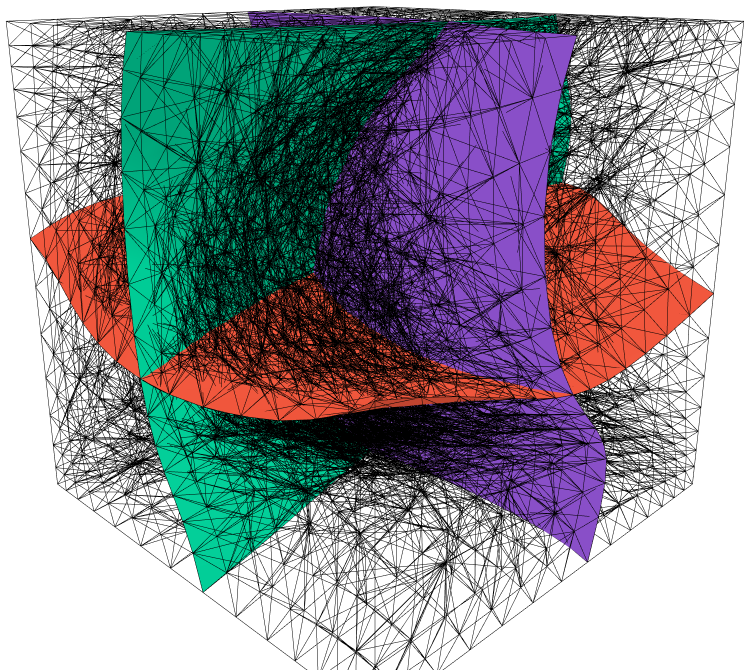} 
        \caption{The reconstructed mesh.}
    \end{subfigure}
    \caption{Mesh reconstruction result of mapping without folding. The reconstruction error is $7.56 \times 10^{-29}$.}
    \label{fig:box2_recon}
\end{figure}
\begin{figure}[t]
    \centering
    \begin{subfigure}[t]{0.32\textwidth}
        \centering
	    \includegraphics[width=\textwidth]{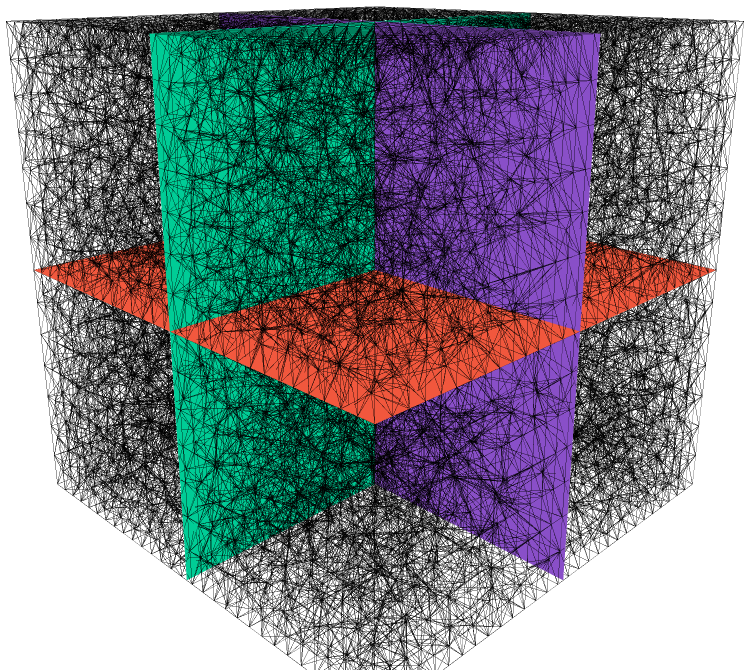}
        \caption{The source mesh.}
    \end{subfigure}
    \begin{subfigure}[t]{0.32\textwidth}
        \centering
        \includegraphics[width=\textwidth]{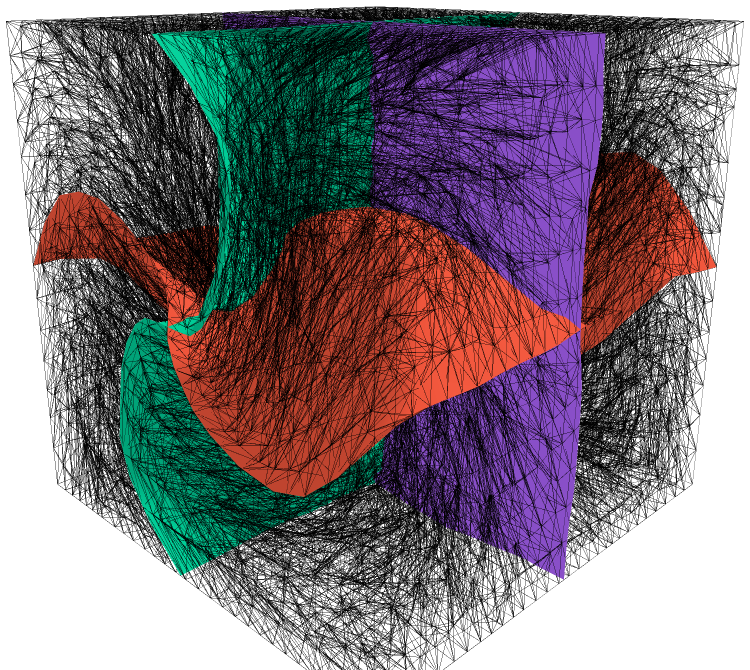} 
        \caption{The target mesh.}
    \end{subfigure}
    \begin{subfigure}[t]{0.32\textwidth}
        \centering
        \includegraphics[width=\textwidth]{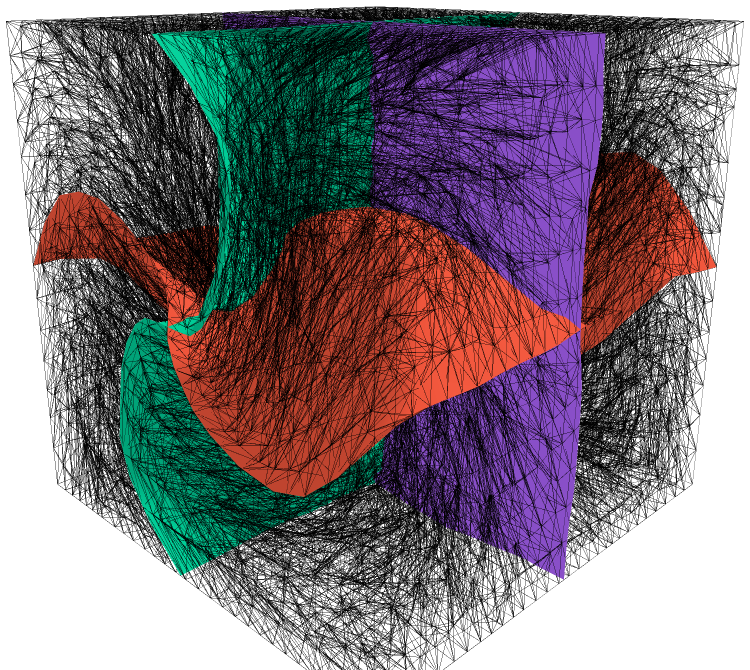} 
        \caption{The reconstructed mesh.}
    \end{subfigure}
    \caption{Mesh reconstruction result of mapping without folding. The reconstruction error is $2.71 \times 10^{-27}$.}
    \label{fig:box8_recon}
\end{figure}
\begin{figure}[t]
    \centering
    \begin{subfigure}[t]{0.32\textwidth}
        \centering
	    \includegraphics[width=\textwidth]{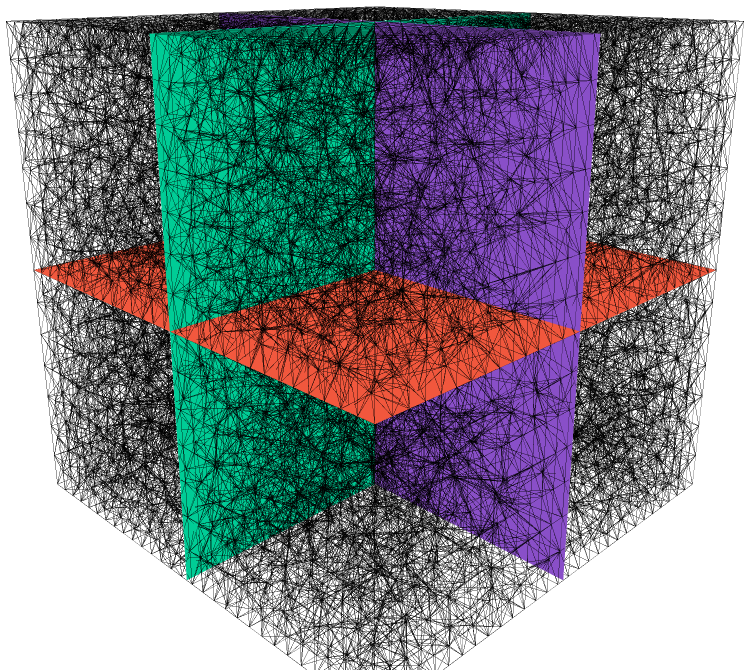}
        \caption{The source mesh.}
    \end{subfigure}
    \begin{subfigure}[t]{0.32\textwidth}
        \centering
        \includegraphics[width=\textwidth]{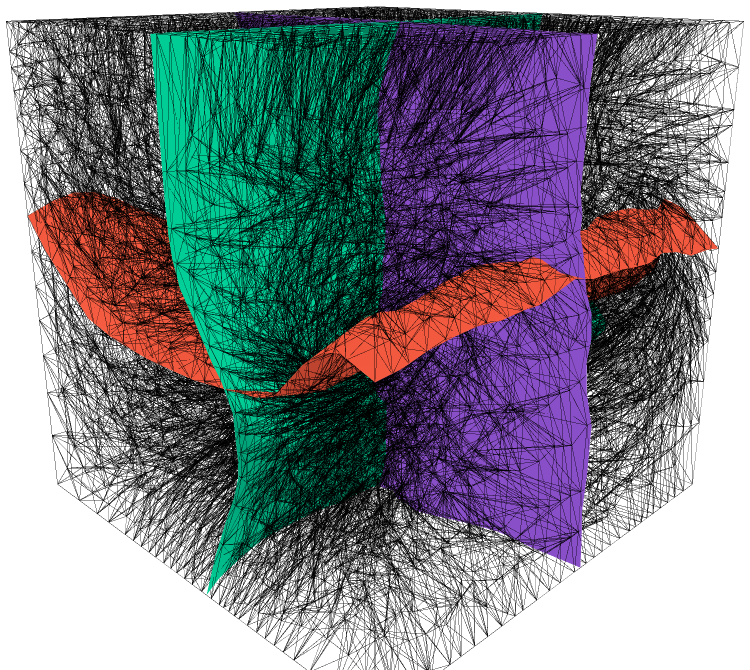} 
        \caption{The target mesh.}
    \end{subfigure}
    \begin{subfigure}[t]{0.32\textwidth}
        \centering
        \includegraphics[width=\textwidth]{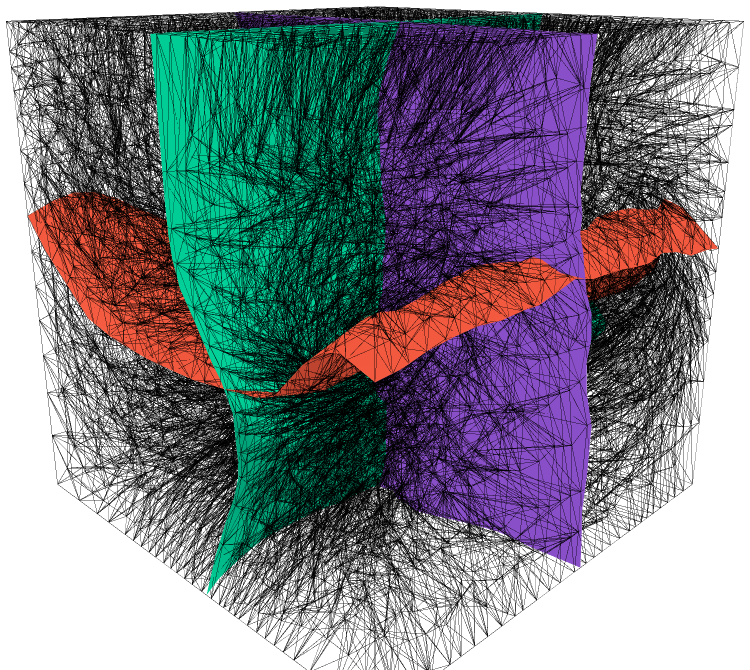} 
        \caption{The reconstructed mesh.}
    \end{subfigure}
    \caption{Mesh reconstruction result of mapping without folding. The reconstruction error is $6.15 \times 10^{-26}$.}
    \label{fig:box9_recon}
\end{figure}
\clearpage
}

In this subsection, we conduct experiments to test the effectiveness of Algorithm~\ref{alg:recon}. Given a 3D bijective mapping, we first compute its corresponding 3DQC. Algorithm~\ref{alg:recon} is then applied to reconstruct the mapping from the given 3DQC. To quantitatively measure the performance, we compute reconstruction errors determined by the $l^2$ norm between the original mappings and the reconstructed mappings.

In the following experiments, we use the same boundary condition that
\begin{equation*}
    f(p \in \mathcal{S}) \in \mathcal{S}, \quad \mathcal{S} \in \textit{BoundarySet}
\end{equation*}
where 
\begin{equation*}
    \textit{BoundarySet} = \bigl\{ \{(i, y, z)\}, \{(x, i, z)\}, \{(x, y, i)\} \mid i \in \{0, 1\}, x, y, z \in \mathbb{R} \bigr\}.
\end{equation*}

Figure~\ref{fig:box2_recon} shows one reconstruction result. Figure~\ref{fig:box2_recon}(a) shows the source mesh. The mesh is deformed by a bijective mapping to the target mesh as shown in Figure~\ref{fig:box2_recon}(b). The corresponding 3DQC is computed, from which the 3D mapping is reconstructed using Algorithm~\ref{alg:recon}. The reconstructed mapping is shown in Figure~\ref{fig:box2_recon}(c), which closely resembles the original mapping. The reconstruction error is $7.56 \times 10^{-29}$.

Figure~\ref{fig:box8_recon} and \ref{fig:box9_recon} show two more reconstruction results, where the corresponding mappings have larger deformations. Again, the reconstructed mappings closely resemble the original mappings. The reconstruction errors are $2.71 \times 10^{-27}$ and $6.15 \times 10^{-26}$, respectively.

\subsection{Applications}
In this section, we present the experimental results for numerical interpolation and sparse modelling of 3D bijective mappings.

\subsubsection{Experimental results for numerical interpolation}

\afterpage{
\clearpage
\begin{figure}[t]
    \centering
    \includegraphics[width=0.73\textwidth]{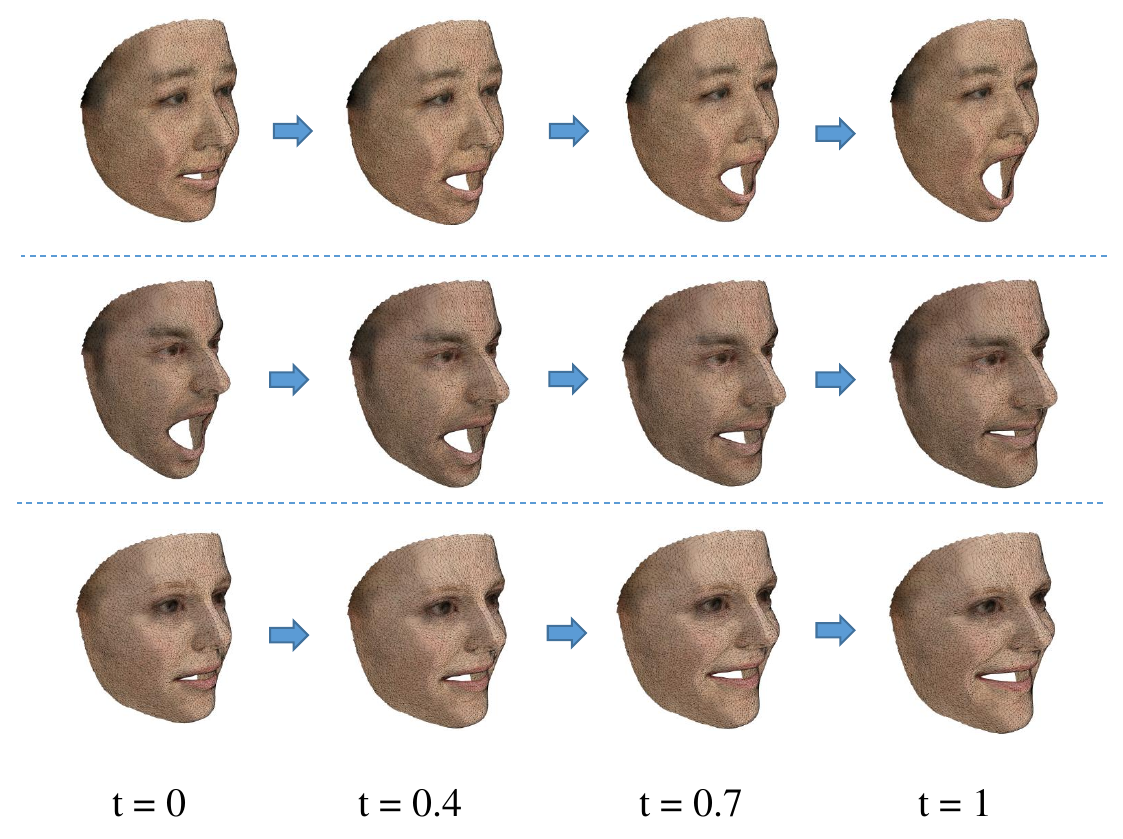}
    \caption{Experimental results for 3D animation interpolation. Each row contains an animation sequence of an object. The objects in the first and last columns are samples generated using 3DMM with different expression parameters, and the second and third columns are the interpolated 3D models for those in the first and last columns. The results show that our algorithm can successfully interpolate between human face meshes to get 3D animation sequences.}
    \label{fig:interp_face}
\end{figure}


\begin{figure}[t]
    \centering
    \includegraphics[width=0.73\textwidth]{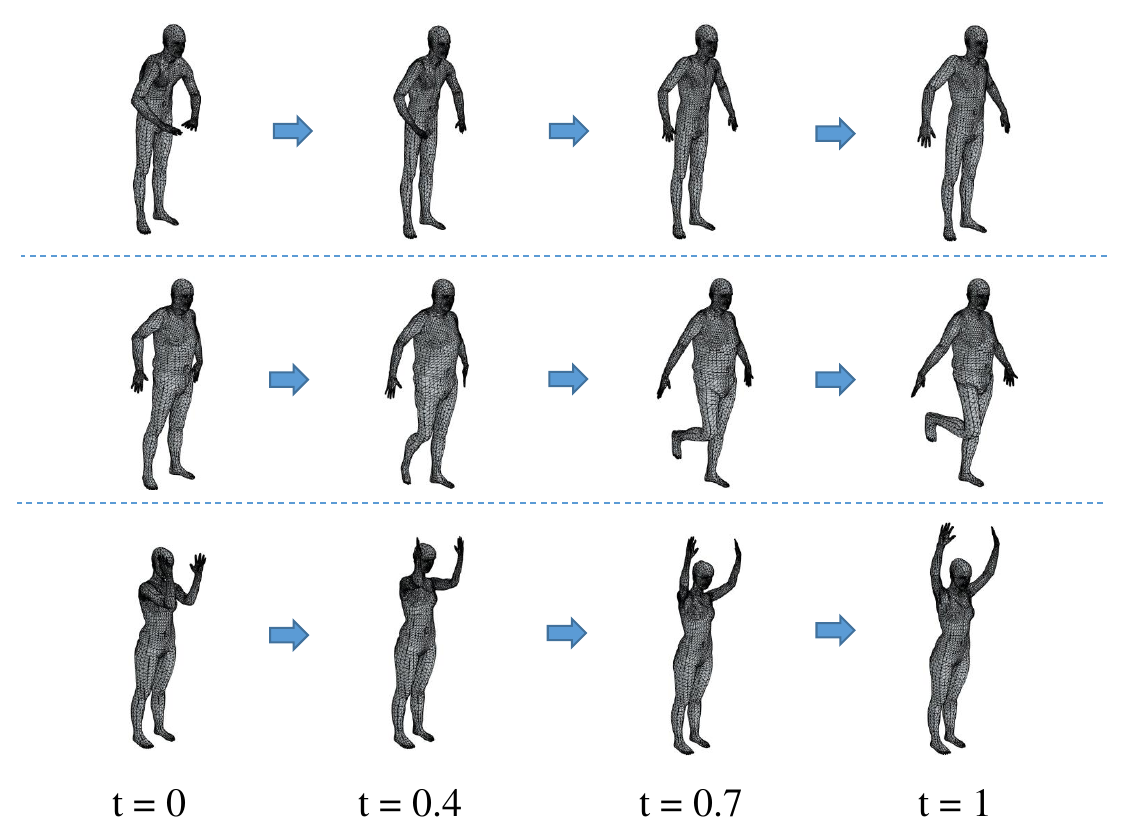}
    \caption{Experimental results for 3D animation interpolation. Each row contains an animation sequence of an object. The objects in the first and last columns are registered human scans sampled from the FAUST dataset, and the second and third columns are the generated 3D models interpolated between those in the first and last columns. The results show that our algorithm can successfully interpolate between large deformation human body scans to get 3D animation sequences.}
    \label{fig:interp_faust}
\end{figure}
\clearpage

\begin{figure}[t]
    \centering
    \includegraphics[width=0.95\textwidth]{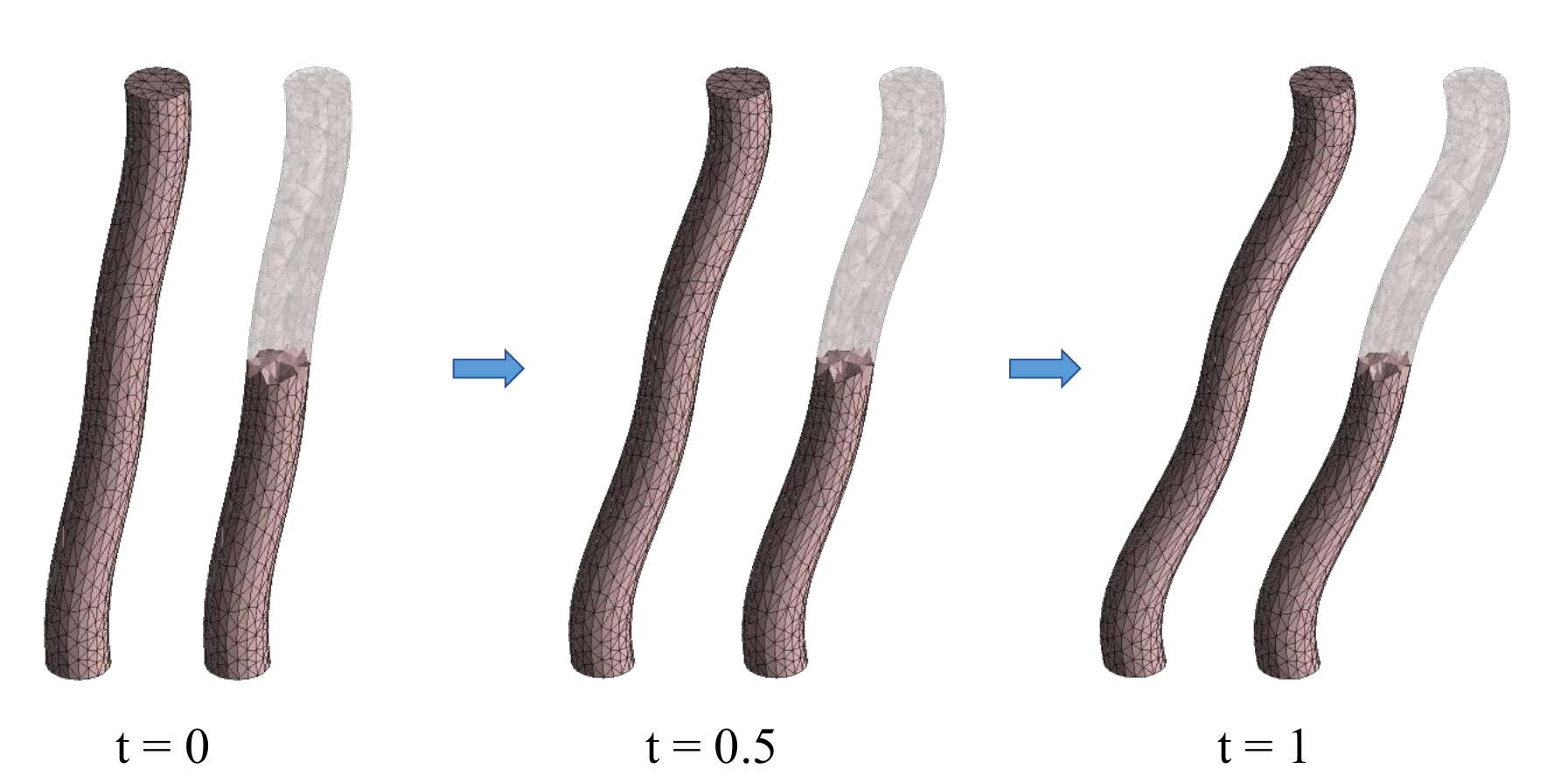}
    \caption{Experimental results for generating deformable domains of a solid blood vessel. For each time step, we show its outer surface and interior tetrahedral structure. The first and last vessels are tetrahedral meshes representing the first and last frames of the domain, and the second vessel is the generated domain interpolated between the first and last vessels. }
    \label{fig:interp_vessel}
\end{figure}

\begin{figure}[t]
  \centering
  \includegraphics[width=\textwidth]{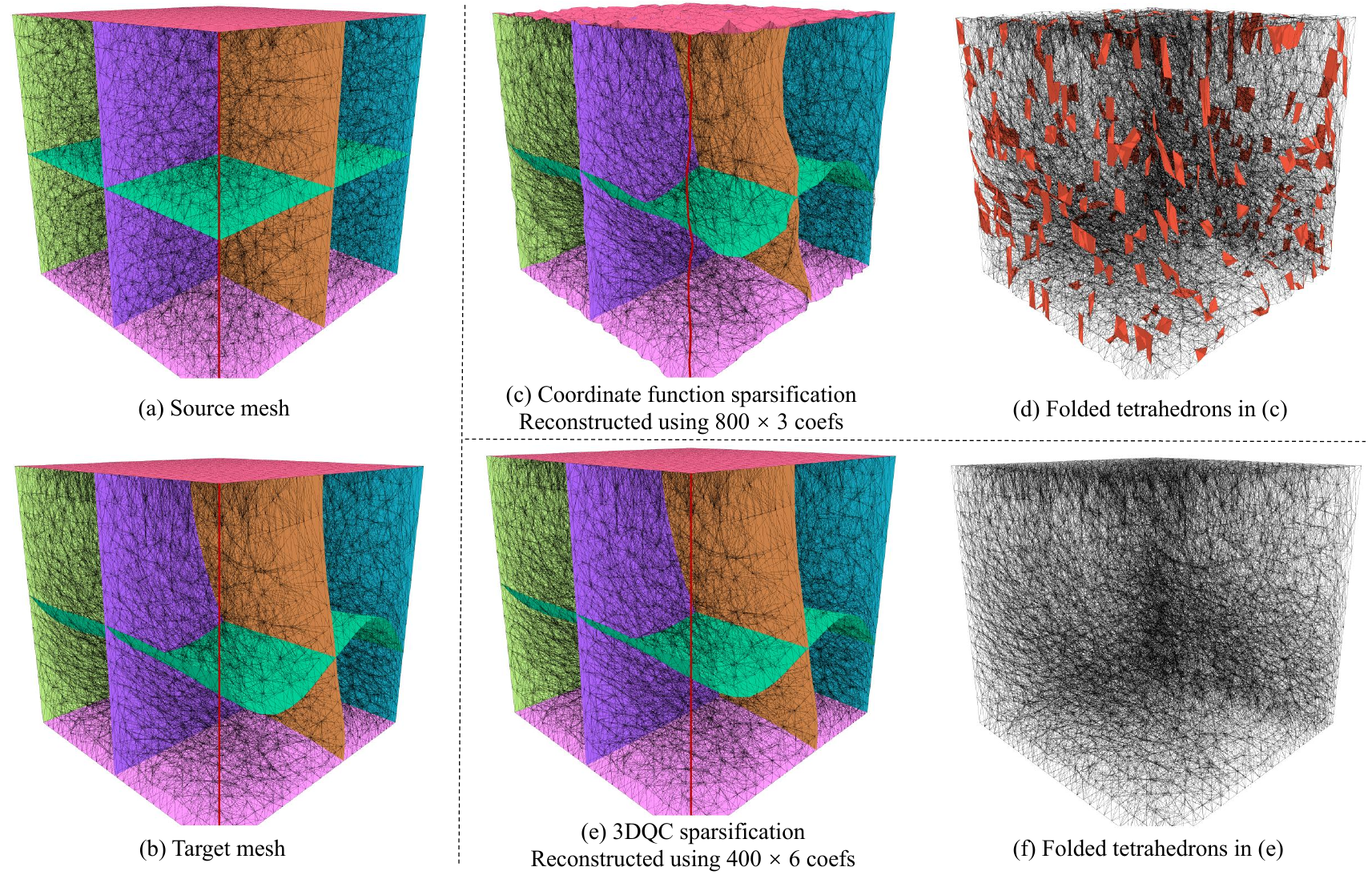}
  \caption{Comparison of coordinate function sparsification and 3DQC sparsification results. (a) and (b) represent the source and target meshes, respectively, forming a mapping. (c) and (e) depict the reconstruction results using 800 and 400 spectral coefficients per component for the coordinate function and 3D representation, respectively. It gives a memory reduction of 90.17\%. Red highlights in (d) and (f) indicate folded tetrahedra in (c) and (e). Notably, (c) contains 448 folded tetrahedra (marked in red in (d)), whereas none are present in (e). The MSE between the meshes in (b) and (c) is $7.95 \times 10^{-5}$, while that between (b) and (e) is $4.70 \times 10^{-5}$, showing a smaller error.}
  \label{box1}
\end{figure}

}

In this subsection, we examine the effectiveness of Algorithm~\ref{code_interp3} for numerical interpolation of 3D bijective mappings. Numerical interpolation of 3D bijective mappings is crucial for applications such as 3D animation or dynamic mesh generation.

In Figure~\ref{fig:interp_face}, we apply Algorithm~\ref{code_interp3} to interpolate human facial surfaces. The human facial meshes are generated using the 3D Morphable Model (3DMM) \cite{3DMM, paysan20093d}. In each row, we show the morphing between two human faces of the same person with different expressions. The initial human face at $t=0$ and the final human face at $t=1$ are embedded in a 3D cube. Landmark-matching quasiconformal registration is computed \cite{zhang2022unifying}. Algorithm~\ref{code_interp3} is used to compute the homotopy between the identity map and the registration map. The interpolated mappings are used to deform the initial facial mesh, which evolves from the initial face to the final face. The three rows show the interpolation of the human faces of three people. Results show that the numerical interpolation scheme generates a homotopy that gives facial surfaces with natural human expression.

We also apply Algorithm~\ref{code_interp3} to interpolate surfaces undergoing larger deformations. Figure~\ref{fig:interp_faust} shows the interpolation results for three human bodies undergoing large motion. In each row, the goal is to interpolate the human body from its initial shape at $t=0$ to its deformed shape at $t=1$. Algorithm~\ref{code_interp3} is used to compute the homotopy between the identity map and the registration map, which generates animation sequences of human body motion. Even with such large deformations, our proposed algorithm performs well, producing natural and smooth morphing between shapes.

In addition, we apply our algorithm to generate dynamic meshes for solving dynamic PDEs. In \cite{zhang2024meshless}, the authors propose an algorithm for blood flow simulations in elastic vessels, where the computational domain is deformable. By applying our proposed method, we can generate dynamic 3D meshes for blood vessels. Utilizing the numerical interpolation of 3DQC, intermediate domains can be generated based on the first and last frames of a blood vessel, as shown in Figure~\ref{fig:interp_vessel}.


\subsubsection{Experimental results for sparse modelling}

In the following experiments, we apply the above sparse modeling scheme to two 3D mappings. The experimental results for the first mapping are shown in Figure~\ref{box1} and ~\ref{fig:sparse_mse_ex1}, while the results for the second mapping are shown in Figure~\ref{box7} and ~\ref{fig:sparse_mse_ex2}.

Figure~\ref{box1} illustrates the sparse modeling of a mapping. The mapping is visualized as the deformation from the source mesh (a) to the target mesh (b). (c) shows the reconstructed result obtained by sparsifying the coordinate function, while (e) presents the reconstructed result using the proposed sparse modeling scheme. To visualize foldings, any folded tetrahedra in (c) and (e) are labeled in red and displayed in (d) and (f), respectively. From the results, we observe that the red edge in (c) deviates significantly from its original position in (b). Additionally, the four colored outer surfaces in (b) become turbulent when the sparse modeling scheme is applied directly to the coordinate function. In contrast, the four outer surfaces and the red edge in (e) remain in their original positions, as the proposed reconstruction algorithm incorporates boundary conditions. In this example, we use 400 coefficients for each channel of the 3DQC and 800 coefficients for each channel of the coordinate functions, yielding a memory reduction of approximately 90.23\%. (d) and (f) demonstrate that, with the same compression rate, the proposed algorithm significantly outperforms direct sparsification of the coordinate function in terms of preserving the bijectivity of the mapping.



The deformation of the mapping shown in Figure~\ref{box7} is significantly larger. In this example, it gives a memory reduction of 90.17\%. Compared with the mesh in (b), in addition to the deviation of the four colored outer surfaces and the red edge, we also observe that the three surfaces inside the cube are noticeably distorted. In contrast, the shape of the three surfaces in (e) is well-preserved, demonstrating the robustness of our algorithm. Furthermore, the results shown in (d) and (f) highlight the robustness of the proposed algorithm in preserving the bijectivity property, even in cases where the mapping undergoes large deformations.

Figure~\ref{fig:sparse_modeling_mse} shows the plots of mean square error (MSE) between the reconstructed mapping and the original mapping versus the number of coefficients used. The left plot corresponds to the mapping in Figure~\ref{box1}, while the right plot corresponds to the mapping in Figure~\ref{box7}. The reference error represents the mean square error between the source and target meshes.  From the plots, we observe that as more high-frequency coefficients are discarded, the error increases. More importantly, the curves demonstrate that the error approaches zero once a small number of coefficients are retained. This indicates that the proposed sparse modeling scheme can achieve a high compression ratio while maintaining a low reconstruction error.

In conclusion, the experimental results demonstrate the superiority of our proposed sparse modeling scheme in several aspects. First, it proves effective for compressing mappings defined on unstructured meshes. Second, the scheme allows for the incorporation of boundary conditions, which is crucial for preserving important geometric features. Lastly, it achieves sparsification of mappings with a high compression ratio, making it a powerful tool for efficient data representation.

\afterpage{
\clearpage
\begin{figure}[t]
  \centering
  \includegraphics[width=\textwidth]{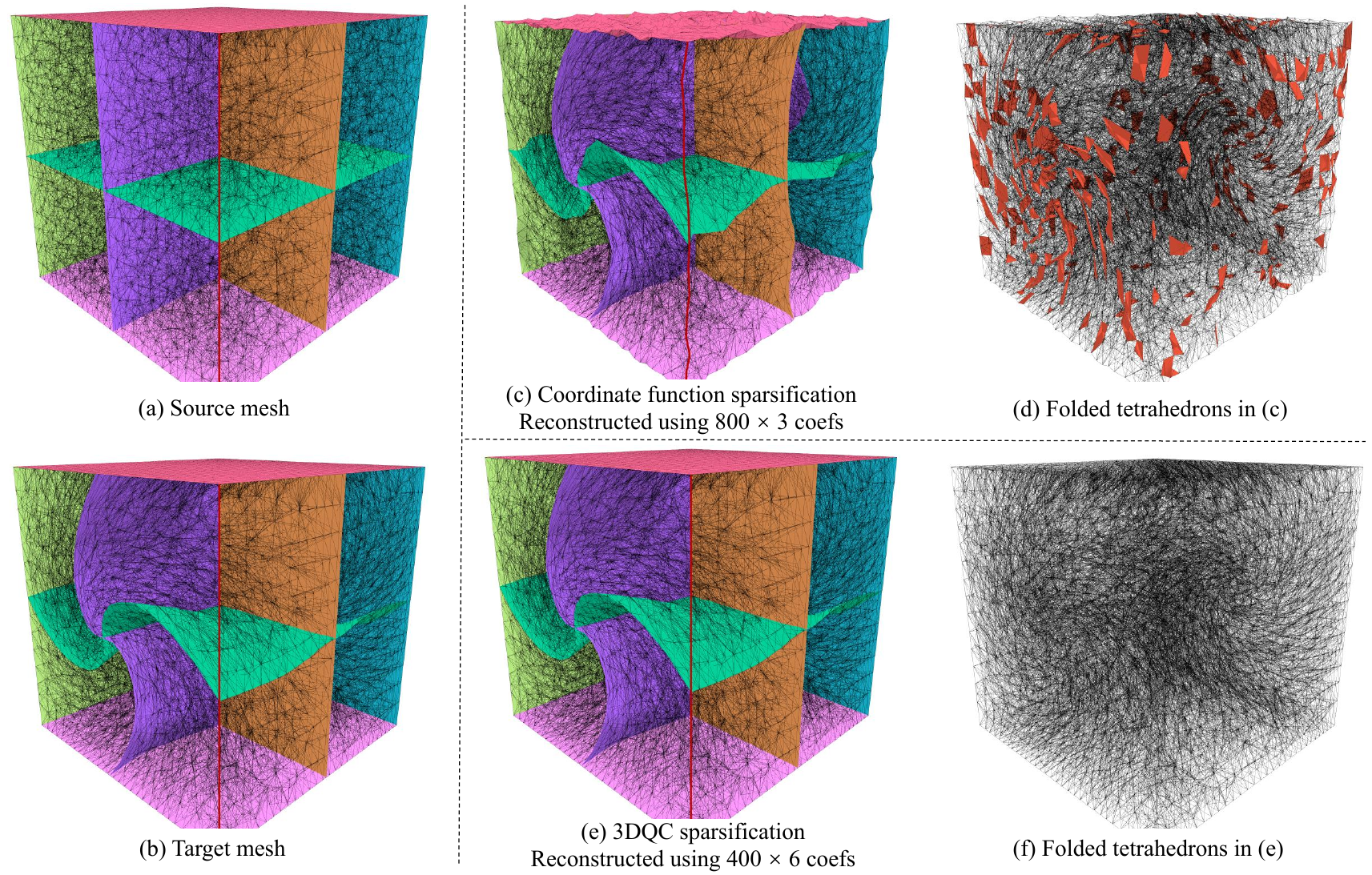}
  \caption{(a) and (b) represent the source and target meshes, respectively, forming a mapping. (c) and (e) depict the reconstruction results using 800 and 400 spectral coefficients per component for the coordinate function and 3D representation, respectively. It gives a memory reduction of 90.23\%. Red highlights in (d) and (f) indicate folded tetrahedra in (c) and (e). Notably, (c) contains 471 folded tetrahedra (marked in red in (d)), whereas none are present in (e). The MSE between the meshes in (b) and (c) is $4.65 \times 10^{-4}$, while that between (b) and (e) is $8.37 \times 10^{-6}$, showing a significantly smaller error.}
  \label{box7}
\end{figure}

\begin{figure}[t]
    \centering
    \begin{subfigure}[t]{0.47\textwidth}
        \centering
	    \includegraphics[width=\textwidth]{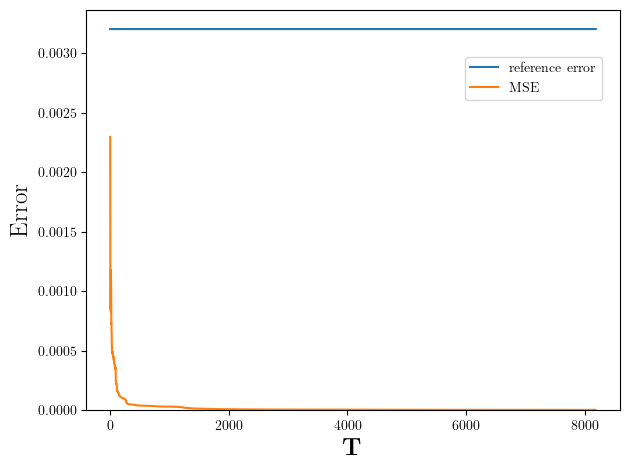}
        \caption{Error between mappings in Figure~\ref{box1}}
        \label{fig:sparse_mse_ex1}
    \end{subfigure}
    \begin{subfigure}[t]{0.47\textwidth}
        \centering
        \includegraphics[width=\textwidth]{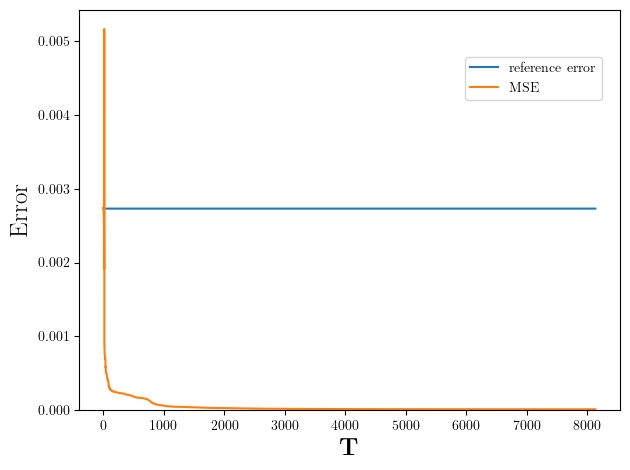} 
        \caption{Error between mappings in Figure~\ref{box7}}
        \label{fig:sparse_mse_ex2}
    \end{subfigure}
    \caption{Chart for MSE visualization for the errors in the two experiments of sparse modeling. The reference error is the MSE between the source mesh and the target mesh. The MSE line represents the MSE between the target mesh and the reconstructed mesh after discarding a certain number of frequency components. $\mathbf{T}$, which represents the number of retained frequency components, is shown on the x-axis.}
    \label{fig:sparse_modeling_mse}
\end{figure}
\clearpage
}

\section{Conclusion}\label{sec:conclusion}
In this paper, we introduced the concept of 3D quasiconformality (3DQC), which effectively captures and represents the local dilation and distortion of 3D mappings. Building on this concept, we developed a reconstruction algorithm that allows us to process mappings directly in the space of 3DQCs.

We validated the proposed 3DQC and reconstruction algorithm through extensive experiments. By reconstructing randomly generated discrete mappings on tetrahedral meshes with high accuracy, we demonstrated the precision of the reconstruction algorithm and the ability of 3DQC to encode distortion information effectively.

Additionally, we explored two applications enabled by 3D quasiconformality and the reconstruction algorithm. First, we proposed a numerical interpolation method for 3D bijective mappings, which generates smooth and bijective transitions between two bijective mappings. Second, we developed a sparse modeling framework for mappings on unstructured meshes, achieving high compression ratios for mappings with smooth deformations while allowing the incorporation of boundary conditions for better control.

For future work, we aim to improve the computational efficiency of the algorithm and extend its applicability to more complex deformations. We also plan to investigate potential applications of 3D quasiconformality in areas such as 3D image registration, image segmentation, and the analysis of deformations in medical imaging.

\bibliographystyle{IEEEtran}
\bibliography{mylib}

\end{document}